\documentclass[acmsmall,screen,nonacm]{acmart}

\usepackage{amsmath}
\usepackage{graphicx}
\usepackage[ruled,linesnumbered,noend]{algorithm2e}
\usepackage{comment}
    \includecomment{cav23} 
    \excludecomment{full} 
    \excludecomment{ver0204}
    \excludecomment{blue}
\usepackage{tabu}
\usepackage{multirow}
\usepackage{float}
\usepackage{wrapfig}
\usepackage{longtable}
\usepackage{paralist}
\usepackage{stmaryrd} 
\usepackage{listings}
\usepackage[disable,colorinlistoftodos,textsize=footnotesize]{todonotes}

\usepackage{tikz}
\usetikzlibrary{automata, positioning, arrows, shapes.geometric}
\tikzset{
->, 
>=stealth', 
node distance=2cm, 
every state/.style={thick}, 
initial text=$ $, 
}

\lstdefinelanguage{affprob}
{
morekeywords={angel,demon, choice, prob(0.6), prob(0.5), if, then, else, fi, 
while, do, od, 
true, false, and, or, skip, sample},
sensitive = false
}

\usepackage{apxproof}
\theoremstyle{plain}
\newtheoremrep{mytheorem}{Theorem}[section]
\newtheoremrep{myproposition}[mytheorem]{Proposition}
\newtheoremrep{mylemma}[mytheorem]{Lemma}
\theoremstyle{definition}
\newtheoremrep{myexample}[mytheorem]{Example}
\newtheoremrep{myremark}[mytheorem]{Remark}
\newtheoremrep{mydefinition}[mytheorem]{Definition}

\usepackage{array}
\newcolumntype{C}[1]{>{\centering\let\newline\\\arraybackslash\hspace{0pt}}m{#1}}
\newcolumntype{R}[1]{>{\raggedleft\let\newline\\\arraybackslash\hspace{0pt}}m{#1}}
\newcolumntype{L}[1]{>{\raggedright\let\newline\\\arraybackslash\hspace{0pt}}m{#1}}

\newcommand{\mynext}{\mathbb{X}}

\newcommand{\sem}[1]{\llbracket #1 \rrbracket}
\newcommand{\br}[1]{\{ #1 \}}
\newcommand{\seqOf}[1]{(#1)_{t=0}^\infty}

\newcommand{\bbE}{\mathbb{E}}

\newcommand{\bbN}{\mathbb{N}}

\newcommand{\bbP}{\mathbb{P}}

\newcommand{\bbR}{\mathbb{R}}

\newcommand{\bbX}{\mathbb{X}}

\newcommand{\calD}{\mathcal{D}}

\newcommand{\calF}{\mathcal{F}}
\newcommand{\calG}{\mathcal{G}}

\newcommand{\calM}{\mathcal{M}}

\newcommand{\calT}{\mathcal{T}}

\newcommand{\calX}{\mathcal{X}}
\newcommand{\calY}{\mathcal{Y}}
\newcommand{\calZ}{\mathcal{Z}}

\newcommand{\Poly}{\mathsf{Poly}}
\newcommand{\pneg}{p_{\mathsf{neg}}}

\newcommand{\Sinf}{S_{\mathsf{Inf}}}
\newcommand{\Sfin}{S_{\mathsf{Fin}}}
\newcommand{\Sany}{S_{\mathsf{Any}}}
\newcommand{\Srem}{S_{\mathsf{Rem}}}
\newcommand{\Oinf}[1][t]{\Omega_{\mathsf{Inf},#1}}
\newcommand{\Ofin}[1][t]{\Omega_{\mathsf{Fin},#1}}
\newcommand{\Oany}[1][t]{\Omega_{\mathsf{Any},#1}}
\newcommand{\myvec}[1]{\boldsymbol{#1}}
\newcommand{\mychar}{{\bf 1}}

\newcommand{\inv}{\mathsf{Inv}}
\newcommand{\pt}{\mathsf{pt}}

\newcommand\xqed[1]{%
  \leavevmode\unskip\penalty9999 \hbox{}\nobreak\hfill
  \quad\hbox{#1}}
\newcommand\demo{\xqed{$\triangle$}}

\usepackage{subfig}
\usepackage{stfloats}
\usepackage{footmisc}

\usepackage[all,pdftex]{xy}
\SelectTips{cm}{}
\xyoption{rotate}

\usepackage{changepage} 
\usepackage{calc}       

\newlength{\pseudowd}
\newcommand{\Strong}[2]{\textsc{S}$_{#1}$\,#2s}
\newcommand{\Lazy}[2]{\textsc{L}$_{#1}$\,#2s}
\newcommand{\Fail}{\textsc{F}}
\newcommand{\NA}{--}

\usepackage{booktabs}
\usepackage{array}
\usepackage{longtable}
\usepackage{ragged2e}
\usepackage{xspace}
\usepackage{xurl}
\AtBeginDocument{%
  }
\begin{document}
\title{Weakly Non-Negative Supermartingales for Omega-Regular Verification}

\author{Toru Takisaka}
\affiliation{%
  \institution{University of Electronic Science and Technology of China}
  \city{Chengdu}
  \country{China}}
\email{takisaka@uestc.edu.cn}

\author{Hongjie Qing}
\affiliation{%
  \institution{University of Electronic Science and Technology of China}
  \city{Chengdu}
  \country{China}
}
\email{202611080124@std.uestc.edu.cn}

\author{Libo Zhang}
\affiliation{%
 \institution{The University of Auckland}
 \city{Auckland}
 \country{New Zealand}}
\email{lzha797@aucklanduni.ac.nz}
\begin{abstract}
Martingale-based methods are central to probabilistic program verification, but strong global non-negativity requirements can exclude simple certificates from tractable template classes. Relaxing this requirement enlarges the search space for automated synthesis, but naive relaxations are unsound in the probabilistic setting. We introduce \emph{lazy Streett supermartingales} and their lexicographic extension, showing that weak non-negativity can nevertheless be used soundly to certify almost-sure satisfaction of $\omega$-regular properties with polynomial templates under a broad class of sampling distributions, including all bounded-support distributions. This extends prior weakly non-negative methods from termination to general $\omega$-regular verification. We further give a compositional account of lexicographic certificates in terms of one-dimensional ones. Experiments on 170 polynomial probabilistic-program benchmarks show increases of 20.0–23.5 percentage points in verification success over the strongly non-negative baseline.
\end{abstract}
 \settopmatter{printacmref=false}
 \setcopyright{none}
 \renewcommand\footnotetextcopyrightpermission[1]{}%

\maketitle

\section{Introduction}\label{sec:introduction}

\paragraph{Background: martingale-based methods and their algorithmic applicability}
Verification of probabilistic programs has been actively studied in recent years.
\emph{Martingale-based methods} are among the central techniques for such verification
tasks. The basic idea is to find a real-valued function over program states---typically
a supermartingale-like certificate---whose value decreases, or is otherwise controlled,
in expectation along program executions.
Such a function can serve as a certificate for properties such as almost-sure
termination~\cite{AgrawalCP18,TakisakaZWL24,ChatterjeeGNZZ23,chakarov2013probabilistic,Huang0CG19,Kenyon-RobertsO21,mciver2017new}, recurrence~\cite{ChakarovVS16}, and, more recently, almost-sure satisfaction of
$\omega$-regular specifications~\cite{AbateGR24,HenzingerMSZ25,KuraU26}.
For example, a \emph{ranking supermartingale} (RSM)~\cite{chakarov2013probabilistic} is a non-negative real-valued function whose value decreases by
a positive constant in expectation across every transition from a non-terminal state. Such a function witnesses that the underlying
probabilistic program is almost surely terminating.

A recurring theme in these developments is how to improve the \emph{algorithmic applicability} of martingale-based methods. Synthesis algorithms typically
restrict their search space to a tractable subclass of supermartingales, such as
linear or polynomial ones. Hence, even if a program admits a supermartingale
that witnesses the target property, a synthesis algorithm may fail to find one
if all such witnesses lie outside its prescribed search space. By
algorithmic applicability, we mean the class of programs for which such
a synthesis procedure can automatically construct a witnessing
supermartingale.

\newsavebox{\firstexloc}
\begin{lrbox}{\firstexloc}
\begin{lstlisting}[mathescape]

$\ell:$
$\quad$
$\quad$
$\quad$
\end{lstlisting}
\end{lrbox}

\lstset{tabsize=2,escapechar=&}
\newsavebox{\firstex}
\begin{lrbox}{\firstex}
\begin{lstlisting}[mathescape]
$y\leftarrow 0$;
while $x \geq 0$ do
$y \leftarrow y+1$;
$x \leftarrow x-y$
od
\end{lstlisting}
\end{lrbox}

\begin{wrapfigure}[6]{r}{0.25\textwidth}
\vspace{-1em}
\centering
\scalebox{0.8}{
\usebox{\firstexloc}
\hspace{0.1cm}
\usebox{\firstex}
}
\end{wrapfigure}

From this viewpoint, one useful way to improve algorithmic applicability is to relax the non-negativity condition imposed on supermartingales.
For example, a linear function $\eta(x,y) = x$ is a ``weakly non-negative ranking function'' that witnesses termination of the non-probabilistic program on the right.
Focusing on the value of $\eta$ at location $\ell$, we see that it satisfies the usual requirements of a ranking function---the non-probabilistic counterpart of an RSM---except for a mild violation of global non-negativity: the value of $\eta$ decreases by at least one in each loop iteration and is non-negative whenever the loop guard $x \geq 0$ is satisfied, \emph{but} it can take an arbitrarily negative value once the loop guard is violated. 
In fact, the program also admits a globally non-negative ranking function $\eta'(x,y) = \max\{x+1, 0\}$: but this function is not linear, and hence can be harder than $\eta$ for synthesis algorithms to find.
Weak non-negativity has been actively studied in the context of \emph{lexicographic extensions} of ranking functions and supermartingales, which provide another well-known way to improve algorithmic applicability. Recent work reports experimentally that weak non-negativity in lexicographic RSMs improves the verification success rate by 15.6\%~\cite{TakisakaZWL24}.

\newsavebox{\secondexloc}
\begin{lrbox}{\secondexloc}
	\begin{lstlisting}[mathescape]
    
$\ell:$
$\quad$
$\quad$
$\quad$
$\quad$
$\quad$
\end{lstlisting}
\end{lrbox}

\lstset{tabsize=2,escapechar=\&}
\newsavebox{\secondex}
\begin{lrbox}{\secondex}
	\begin{lstlisting}[mathescape]	
$x\leftarrow 0$;$y\leftarrow 1$;
while $x \geq 0$ do
  $y \leftarrow  y+1$;
  if $\mathbf{prob}(2^{-y})$ then
    $x \leftarrow -2^y$
  fi
od
\end{lstlisting}
\end{lrbox}

\begin{wrapfigure}[9]{r}{0.26\textwidth}
\vspace{-1em}
\centering
\scalebox{0.8}{
    \usebox{\secondexloc}
        \hspace{0.1cm}
    \usebox{\secondex}
 }
\caption{A counterexample.} 
\label{fig:CounterEx}
\end{wrapfigure}

\paragraph{Challenge: subtlety of weak non-negativity in martingales}
This weakening of non-negativity is known to be substantially more delicate in the probabilistic setting:
a naive transfer of the non-probabilistic intuition can lead to an unsound martingale notion.
The program on the right is a known counterexample~\cite{ChatterjeeGNZZ23,TakisakaZWL24}.
The function $\eta(x,y)=x$ satisfies properties analogous to those in the non-probabilistic example above:
in each loop iteration, the value of $\eta$ decreases by one in expectation, namely from $0$ to $-2^y \times 2^{-y} = -1$, and it is non-negative whenever the loop guard $x \geq 0$ is satisfied.
However, the program does not terminate almost surely: a simple calculation shows that the termination probability is over-approximated by $2^{-2} + 2^{-3} + \cdots = 0.5$.
Thus, such a ``weakly non-negative RSM'' is not sound in general as a certificate of almost-sure termination.
The issue arises because a supermartingale condition concerns an \emph{expected value} after a program transition, so negative values can be exploited in a spurious way.

Nevertheless, prior work has shown that this obstruction can sometimes be overcome.
A recent paper~\cite{TakisakaZWL24} observed that synthesis algorithms search only within restricted classes of programs and RSMs, and that these restrictions can rule out counterexamples such as the one above.
For example, probabilistic branching with a variable probability, such as $\mbox{`\textbf{if}} \, \mbox{\textbf{prob$(x)$}'}$, plays a crucial role in the counterexample above, whereas existing synthesis algorithms often allow only $\mbox{`\textbf{if}} \, \mbox{\textbf{prob$(q)$}'}$ with a constant $q \in [0,1]$.
It has been shown in~\cite{TakisakaZWL24} that weakly non-negative RSMs, as well as their lexicographic extension called \emph{lazy lexicographic RSMs}, are sound under suitable assumptions.

The idea developed in~\cite{TakisakaZWL24}---that weakly non-negative martingales are unsound in general but can become sound under restrictions induced by synthesis algorithms---is still far from fully explored.
Their soundness proof is formulated for a reachability property, namely termination; it assumes linear templates; and it relies on specific regularity assumptions on the sampling distributions, called \emph{well-behavedness}.
Moreover, the proof is rather intricate: it directly analyzes a multi-dimensional lexicographic object and uses a delicate argument showing how negative values can be ``fixed'' without destroying the ranking condition.
This makes it difficult to adapt their results to different settings, even when such an adaptation is in principle possible.

\paragraph{Contributions}
This paper pushes the idea of~\cite{TakisakaZWL24} further and significantly generalizes their results.
We propose a novel martingale notion of \emph{lazy Streett supermartingales} (lazy SSMs, Definition~\ref{def:LazySSMMap}) and its lexicographic extension, the notion of \emph{lazy lexicographic SSMs} (lazy LexSSMs, Definition~\ref{def:LazyLexSSMMap}).
This notion combines two existing martingale notions: lazy (lexicographic) RSMs~\cite{TakisakaZWL24}, RSM notions with the most liberal non-negativity condition to date;
and \emph{Streett supermartingales} (SSMs)~\cite{AbateGR24,KuraU26}, recently proposed martingale certificates for almost-sure satisfaction of $\omega$-regular properties.

Table~\ref{tab:lazy-lex-comparison} compares the scopes of lazy RSMs and lazy SSMs.
First, lazy RSMs certify only almost-sure termination of probabilistic programs, whereas lazy SSMs certify almost-sure satisfaction of $\omega$-regular properties.
Second, the soundness result in~\cite{TakisakaZWL24} assumes lazy RSMs to be linear;
we prove soundness for polynomial lazy SSMs, a common template class in recent verification algorithms~\cite{chatterjee2016termination,TakisakaOUH21,AbateGR24,HenzingerMSZ25}.
Third, the existing result requires program arithmetic to be linear: that is, a function $f$ in a program line $x \leftarrow f(\myvec{x})$ must be linear.
We allow $f$ to be any measurable function.
Fourth, the existing result requires sampling distributions to be \emph{well-behaved} (Definition~\ref{def:wellBehavedOriginal});
we replace this assumption with what we call \emph{relative well-behavedness} (Definition~\ref{def:weakWellBehaved}).
This last generalization becomes essential for polynomial extension of soundness, as discussed below.

The soundness proof developed in this paper requires reasoning \emph{that is
fundamentally different from the existing proof for lazy RSMs}, mainly because
we move from linear to polynomial templates.  The extension from reachability to
$\omega$-regular properties is conceptually handled by the Streett-supermartingale
framework, but the polynomial-template extension raises a separate and more
subtle technical challenge: the well-behavedness condition used in the existing
lazy RSM proof becomes too restrictive in the polynomial setting.  It even rules
out uniform distributions over intervals, thereby severely limiting the practical
relevance of the resulting soundness theorem.  To overcome this obstacle, we
develop a different proof strategy based on 
\emph{relative well-behavedness} (Definition~\ref{def:weakWellBehaved}).
We show that every bounded-support distribution is relatively well-behaved for
polynomial templates, which yields a practically meaningful class of applicable
probabilistic programs.  We give an overview of this phenomenon in
Section~\ref{subsec:polySoundness}.

\begin{table}[t]
  \centering
  \caption{Comparison of lazy lexicographic supermartingale techniques}
  \label{tab:lazy-lex-comparison}
  \footnotesize
  \renewcommand{\arraystretch}{1.08}
  \setlength{\tabcolsep}{6pt}
  \begin{tabular}{@{}p{0.35\textwidth}p{0.27\textwidth}p{0.27\textwidth}@{}}
    \toprule
     & Lazy LexRSM~\cite{TakisakaZWL24} & Lazy LexSSM (this paper) \\
    \midrule
    Certified property
      & Reachability
      & $\omega$-regular properties \\
    Template class
      & Linear 
      & Polynomial \\
    Program updates
      & Linear updates
      & Measurable updates \\
Assumption on sampling distributions
      & Well-behavedness
      & Relative well-behavedness \\
    \bottomrule
  \end{tabular}
\end{table}

Our technical development also introduces several ideas that make the analysis simpler and better structured.
First, we isolate the main technical difficulty in the soundness proof for \emph{one-dimensional} lazy SSMs, and then derive soundness of their lexicographic extension as a simple corollary.
The key observation is that a lexicographic SSM is simply a sequence of one-dimensional SSMs that jointly certify the target $\omega$-regular property; we give an overview in Section~\ref{subsec:fromOneDimToLex}.
As a byproduct, we present a ``lazy'' variant of \emph{lexicographic progress-measure supermartingales}~\cite{KuraU26}, another form of multi-dimensional supermartingale.

Second, we formulate our results over a new semantic model, which we call
\emph{program MDPs}, rather than directly over the conventional \emph{probabilistic control flow graph} formalism.
The soundness of lazy RSMs implicitly relies on additional properties of pCFGs that do not hold for general MDPs, and the pCFG formalism makes it difficult to identify which of these properties are essential.
Program MDPs abstract the relevant properties of pCFGs as explicit axioms.
This gives a clearer account of the scope of our soundness results, while also allowing us to use these axioms to develop sound weakly non-negative martingale notions.

To examine the practical benefit of weak non-negativity, we performed experiments over 170 synthetic benchmarks of polynomial probabilistic programs (Section~\ref{sec:evaluation}).
Results show that our synthesis algorithm based on lazy lexicographic SSMs achieved verification rates 20.0--23.5 percentage points higher than those of the strongly non-negative 
counterpart~\cite{KuraU26}.

\medskip

The rest of the paper is organized as follows. Section~\ref{sec:summary} 
summarizes the main technical points: the transition to relatively well-behaved distributions, and multi-dimensional lifting of the soundness proof of one-dimensional SSMs.
After preliminaries in Section~\ref{sec:preliminaries}, 
we give the technical details about well-behavedness and its ``relative'' variant in Section~\ref{sec:RelativeWellBehaved}; 
define lazy SSMs and prove their soundness in Section~\ref{sec:LazySSMSoundness}; 
extend the soundness to multi-dimensional SSM variants in Section~\ref{sect:fromOneToLex}. 
We summarize a synthesis algorithm for lazy lexicographic SSMs in Section~\ref{sec:algorithm}, and show experimental results in Section~\ref{sec:evaluation}. 
After related work in Section~\ref{sec:relWorks}, we conclude in Section~\ref{sec:conclusion}.

\section{Summary of Technical Key Points}\label{sec:summary}
\subsection{From Well-Behaved to Relatively Well-Behaved Distributions}\label{subsec:polySoundness}
We give a brief explanation by example about what well-behavedness is, how it matters for the soundness of linear lazy RSMs, how a similar argument essentially fails in the polynomial case, and how we come up with our alternative condition, relative well-behavedness.

\paragraph{What are well-behaved distributions?}
As an example of a well-behaved distribution, we take $\mathrm{Unif}[0,1]$, i.e., the uniform distribution over $[0,1]$. 
Below we give an informal explanation in the context of program execution; a formal definition is given in Definition~\ref{def:wellBehavedOriginal}.

Suppose we have a probabilistic program that involves a program line `$\ell: \ x \sim \mathrm{Unif}[0,1]$', which randomly updates a variable $x$ to a value sampled from $\mathrm{Unif}[0,1]$. 
Also suppose the program involves two variables $(x,y)$, for simplicity. 
Then an execution of `$\ell: \ x \sim \mathrm{Unif}[0,1]$' will perform the following update to the program state $(\ell,x,y)$, where $\ell'$ is the successor location of $\ell$:
\begin{align*}
    (\ell, x, y) \xrightarrow{\quad \ell: \ x \sim p \quad} (\ell', x', y), \qquad \mbox{where} \qquad x' \sim \mathrm{Unif}[0,1].
\end{align*}
Now suppose a linear function $\eta_{\ell'}(x,y)$ is associated to $\ell'$ as a part of an RSM: Here, an RSM is a function $\eta$ over program states, or equivalently, a collection $(\eta_\ell)_{\ell\in L}$ of functions over variables for each program location. Let us say $\eta_{\ell'}(x,y) = x-y$ for example: 
Then after the above execution, the value of $\eta_{\ell'}$ will be %
updated to
\begin{align*}
    \eta_{\ell', y}(x') = x' - y, \qquad \mbox{where} \qquad x' \sim \mathrm{Unif}[0,1].
\end{align*}
Here, we are interested in two relevant values: One is $\mathrm{Pr}(\eta_{\ell', y}<0)$, the probability of observing $\eta_{\ell', y}(x') <0$; another is $\int_0^1 (\eta_{\ell', y})^- \ dx$, the ``mean shortfall'' of $\eta$ ($f^-$ stands for $\max\{-f, 0\}$). 
These values depend on what $y \in \bbR$ is, but we observe the following relationship: when $\mathrm{Pr}(\eta_{\ell', y}<0) \in (0,1)$, i.e., when $y \in (0,1)$, we have
\[
\int_0^1 (\eta_{\ell', y})^- \ dx = \int_0^{y}(y-x) \ dx = \frac{1}{2} \cdot (\mathrm{Pr}(\eta_{\ell', y}<0))^2.
\]
Thus in particular, for any sequence $y_0, y_1, \ldots$ of reals, the following implication holds:
\begin{align}
     \mbox{if} \quad \mathrm{Pr}(\eta_{\ell', y_t}<0) \rightarrow 0 \quad (t\rightarrow \infty), 
     \quad\mbox{then}\quad
     \int_0^1 (\eta_{\ell', y_t})^- \ dx \rightarrow 0\quad (t\rightarrow \infty).\label{eq:convergenceRel}
\end{align}%
The same convergence relationship holds for any $\eta_{\ell'}$ whenever $\eta_{\ell'}$ is linear. 
We call this property the \emph{well-behavedness} of the distribution $\mathrm{Unif}[0,1]$ (precisely speaking, there is also a requirement on the relative convergence speed of $\mathrm{Pr}(\eta_{\ell', y_t}<0)$ and $\int_0^1 (\eta_{\ell', y})^- \ dx$, see Definition~\ref{def:wellBehavedOriginal}). 
Well-behavedness is shown to be a mild requirement: any bounded-support distribution, as well as some unbounded-support ones such as normal distributions, are well-behaved~\cite[Proposition C.6]{TakisakaZWL24arXiv}.

\paragraph{Relation to the soundness of linear lazy RSMs} Intuitively, well-behavedness of a distribution $p$ ensures that linear lazy RSMs cannot do ``ill-exploitation'' of unbounded negative values at a program line `$x \sim p$'. 
From the counterexample program in Figure~\ref{fig:CounterEx}, 
we find the following update pattern of $\eta$ is its crucial aspect:
\begin{enumerate}[(a)]
    \item The probability that $\eta(x,y) =x$ becomes negative upon the $t$-th iteration of the while loop converges to zero as $t$ increases: Such a probability in the $t$-th iteration is $2^{-(t+1)}$.\label{item:keyProp1illExploitation}
    \item The mean shortfall of $\eta(x,y) = x$ remains a strictly positive constant through each iteration of the while loop: In the $t$-th iteration, it is calculated as $2^{t+1} \times 2^{-(t+1)} = 1$.\label{item:keyProp2illExploitation}
\end{enumerate}
Well-behavedness of $p$ claims that $\eta$ never exhibits such an update pattern through an iterative execution of a program line `$x \sim p$' when $\eta$ is a linear function. Indeed, the convergence relationship (\ref{eq:convergenceRel}) says that the conditions (\ref{item:keyProp1illExploitation}) and (\ref{item:keyProp2illExploitation}) cannot be satisfied simultaneously.

\paragraph{How does a similar argument fail in the polynomial case?}
From the explanation above, one would notice that the notion of well-behavedness is tailored to the soundness proof of linear lazy RSMs: 
We consider an arbitrary \emph{linear} $\eta_{\ell'}$ and see if (\ref{eq:convergenceRel}) holds. 
Therefore, if we wish to perform a similar argument for a polynomial lazy RSM, then we need to modify the well-behavedness requirement accordingly, i.e., now we consider an arbitrary \emph{polynomial} $\eta_{\ell'}$ and see if (\ref{eq:convergenceRel}) holds. 

It turns out such a requirement is now too restrictive: Even $\mathrm{Unif}[0,1]$ is not well-behaved when $\eta_{\ell'}$ can be polynomial. 
Indeed, suppose the underlying program has three variables $(x,y,z)$, and let $\eta_{\ell'}(x,y,z) = xy+z$. 
Also define sequences $\seqOf{y_t}$ and $\seqOf{z_t}$ by $y_t=4^t$ and $z_t=-2^t$. 
Then we have a sequence of functions $\eta_{\ell',y_t,z_t}(x) = \eta_{\ell'}(x,y_t,z_t) = 4^t x - 2^t$ that violates (\ref{eq:convergenceRel}): 
One can check 
$\mathrm{Pr}(\eta_{\ell', y_t, z_t}<0) \rightarrow 0$ as $t\rightarrow \infty$ while 
$\int_0^1 (\eta_{\ell', y_t, z_t})^- \ dx = \frac{1}{2}$ for each $t$. 

\paragraph{Our solution: transition to relative well-behavedness}
The above example shows that, if $\eta_{\ell'}$ is polynomial, then there is a sequence $\seqOf{y_t, z_t}$---therefore, a sequence $\seqOf{\ell, x_t,y_t, z_t}$ of program states---under which $\eta_{\ell'}$ violates the condition (\ref{eq:convergenceRel}). 
But this does \emph{not} appear as a problem for the soundness of polynomial lazy RSMs \emph{as long as such a sequence of program states does not appear in a program run}. 
In fact, there is yet another aspect that hinders an appearance of such a program run: 
\begin{align}
    \mbox{We have $ \ \int_0^1 \eta_{\ell', y_t, z_t} \ dx \rightarrow +\infty \ $ as $ \ t \rightarrow \infty$.}\label{eq:tradeoff}
\end{align}
Because $\eta_{\ell'}$ is supposed to be part of a polynomial lazy RSM, 
this means we also have $\eta_\ell(x_t,y_t,z_t) \rightarrow +\infty$ as $t \rightarrow \infty$. 
This cannot happen ``for free'' in a program run because the ranking condition of $\eta$ forces its value to keep decreasing through a program run, in expectation.

Our key findings are twofold: First, this trade-off relationship---any sequence $\seqOf{ \eta_{\ell', y_t, z_t}}$ that violates (\ref{eq:convergenceRel}) necessarily satisfies (\ref{eq:tradeoff})---holds for any polynomial $\eta_{\ell'}$, and not only under $\mathrm{Unif}[0,1]$ but also under any bounded-support distribution $p$. 
We formalize this property as the \emph{relative well-behavedness} of $p$ (Definition~\ref{def:weakWellBehaved}: The notion also comes with a convergence speed requirement, precisely speaking), 
and show any bounded-support distribution satisfies it in the polynomial case (Proposition~\ref{prop:WeakWellBehaved}). 

Second, we find that the relative well-behavedness clears a path to the soundness proof of polynomial lazy RSMs (and more generally, polynomial lazy SSMs), via a totally different reasoning from the linear case. 
The core technical leap is to consider a \emph{fractional-power transform} of a lazy RSM $\eta$, i.e., a function $(\eta^+)^\theta$, as a witness of almost-sure termination: 
By combining the relative well-behavedness condition and the standard Jensen's inequality, we have a core inequality that lets us translate the ranking condition of $\eta$ into that of $(\eta^+)^\theta$ (Lemma~\ref{lem:keyIneq}).

\subsection{New Soundness Proof Strategy of Lexicographic Martingales: SSM Rules All}\label{subsec:fromOneDimToLex}
We summarize an idea of jointly certifying $\omega$-regular properties by multiple supermartingales, and point out its equivalence with the notion of \emph{lexicographic supermartingales}, a well-known multi-dimensional extension of supermartingales. A formal statement of the equivalence and proof are given in Section~\ref{subsec:EquivalenceLex}, taking a \emph{(strongly non-negative) Streett supermartingale} as an example.

Recall a \emph{Streett condition} 
is an $\omega$-regular property characterized by a pair $(\Sfin, \Sinf)$ of sets of program states, called a  \emph{Streett pair} (let us assume $\Sfin$ and $\Sinf$ are disjoint, for simplicity).
A program run, i.e., a sequence $s_0s_1\ldots$ of program states, satisfies the Streett condition
$(\Sfin, \Sinf)$ if it either visits $\Sfin$ only finitely often, or  
visits $\Sinf$ infinitely often: 
\begin{align}
    \overset{\infty}{\forall}t. s_t \not\in \Sfin 
    \quad \mbox{or} \quad
    \overset{\infty}{\exists}t. s_t \in \Sinf.
\end{align}%
Now take any $S' \subseteq S \setminus \Sinf$. 
Then the Streett condition $(\Sfin, \Sinf)$ is implied by %
the conjunction of the following Streett conditions: 
\begin{align}
    (S', \Sinf) \quad , \quad (\Sfin \setminus S', \Sinf \cup S').\label{eq:twoStreettPairs}
\end{align}
In fact, we have the following reasoning:
\begin{itemize}
    \item The Streett pair $(S', \Sinf)$ stands for $ \overset{\infty}{\forall}t. s_t \not\in \Sinf \implies \overset{\infty}{\forall}t. s_t \not\in S' \implies \overset{\infty}{\forall}t. s_t \not\in \Sinf \cup S'$.%
     \item The Streett pair $(\Sfin \setminus S', \Sinf \cup S')$ stands for 
     $ \overset{\infty}{\forall}t. s_t \not\in \Sinf \cup S'
     \implies \overset{\infty}{\forall}t. s_t \not\in \Sfin \setminus S'
      $.
    \item Thus (\ref{eq:twoStreettPairs}) implies $ \overset{\infty}{\forall}t. s_t \not\in \Sinf 
    \implies \overset{\infty}{\forall}t. s_t \not\in S' 
    \implies \overset{\infty}{\forall}t. s_t \not\in \Sinf \cup S'
    \implies \overset{\infty}{\forall}t. s_t \not\in \Sfin \setminus S'$, in particular
     $ \overset{\infty}{\forall}t. s_t \not\in \Sinf 
    \implies \overset{\infty}{\forall}t. s_t \not\in \Sfin$: 
    This is the Streett condition $(\Sfin, \Sinf)$.
\end{itemize}%
It is also easy to see the opposite implication holds when $S' \subseteq \Sfin$: 
Just observe $(\Sfin, \Sinf)$ implies $(\Sfin',\Sinf')$ when $\Sfin \supseteq \Sfin'$ and $\Sinf \subseteq \Sinf'$.%

More generally, for disjoint subsets $S_1, \ldots, S_k$ of $S \setminus \Sinf$, the Streett condition $(\Sfin, \Sinf)$ can be jointly represented by the Streett pairs $(\Sfin^{(1)}, \Sinf^{(1)}),\ldots,(\Sfin^{(k)}, \Sinf^{(k)})$, where
\[
\Sfin^{(i)} = S_i \quad , \quad \Sinf^{(i)} = \Sinf \cup (S_1 \cup \cdots \cup S_{i-1}).
\]
This decomposition suggests an incremental approach for $\omega$-regular verification of probabilistic programs via \emph{Streett supermartingales}\footnote{
The original notion is due to~\cite{AbateGR24}. We consider only the
refined variant introduced in~\cite{KuraU26} as a
\emph{generalized Streett supermartingale} and refer to it simply as
a Streett supermartingale.%
} (SSMs)~\cite{AbateGR24,KuraU26}, or our lazy SSMs. 
An SSM can witness almost-sure satisfaction of Streett conditions, 
so even if a synthesis algorithm fails to find an SSM under $(\Sfin, \Sinf)$, it can still attempt to witness $(S', \Sinf)$ with $S'\subseteq \Sfin$ as large as possible, 
and then try to witness $(\Sfin \setminus S', \Sinf \cup S')$. 
If successful, the algorithm would generate a sequence $\eta_1, \ldots, \eta_k$ of SSMs that witness $(\Sfin^{(1)}, \Sinf^{(1)}),\ldots,(\Sfin^{(k)}, \Sinf^{(k)})$ as above, jointly witnessing $(\Sfin, \Sinf)$.

Our observation is that such a sequence $\eta_1, \ldots, \eta_k$ is nothing but a \emph{lexicographic Streett supermartingale} proposed in~\cite{KuraU26}, and the incremental algorithm above is exactly the standard design of its synthesis algorithm. 
Many other multi-dimensional martingales can be explained in a similar way: \emph{lexicographic ranking supermartingale}~\cite{AgrawalCP18}, \emph{progress-measure supermartingale}~\cite{KuraU26} and its lexicographic extension~\cite{KuraU26} are all explained as sequences of SSMs that jointly verify the target property. 
Based on this observation, we also define our \emph{lazy lexicographic SSM} (Definition~\ref{def:LazyLexSSMMap}) as a sequence of one-dimensional lazy SSMs.%

A notable advantage of this observation is that it significantly simplifies the soundness proof of multi-dimensional martingales. 
In the conventional proof approach, we directly analyze the behavior of martingales over multi-dimensional vectors equipped with a non-standard order, e.g., the lexicographic ordering with a constant gap. This is typically complex, and such a complexity becomes rather excessive for weakly non-negative variants because the proof involves a delicate analysis about potential ``ill-exploitation'' of unbounded negativity. 
Based on our observation, we can now focus on proving the soundness in a one-dimensional case; then the soundness in a multi-dimensional case follows as an almost obvious corollary.

\section{Preliminaries}\label{sec:preliminaries}

\paragraph{Notations}
We let $\bbN$ and $\bbR$ be the sets of non-negative integers and reals, respectively. 
We also let $\overline{\bbR} = \bbR \cup\{+\infty, -\infty\}$, $\bbR_{>0} = (0,+\infty)$, and $\bbR_{\geq 0} = [0,+\infty)$. %
For $\calY \subseteq \calX$, the \emph{characteristic function} of $\calY$ is a function $\mychar_\calY: \calX \to \{0,1\}$ such that $\mychar_\calY(x) = 1$ iff $x \in \calY$. 
For a function $f:\calX \times \calY \to \calZ$ and $y \in \calY$, we write $f_y: \calX \to \calZ$ to denote the function $f_y(x) = f(x,y)$. 
The value of a vector $\boldsymbol{x}\in \bbR^n$ at the $i$-th index is denoted by %
$x_i$, and its remainder $(x_1, \ldots, x_{i-1}, x_{i+1},\ldots,x_n) \in \bbR^{n-1}$ is denoted by $\myvec{x}_{-i}$. 
For a real-valued function $f: X \to \bbR$, the \emph{positive} and \emph{negative parts} of $f$ are functions $f^+, f^-: X \to \bbR$ such that $f^+(x) = \max\{f(x), 0\}$ and $f^-(x) = \max\{-f(x), 0\}$ (hence we have $f = f^+-f^-$).

The set of all probability distributions over a set $\Omega$ is denoted by $\calD(\Omega)$, where we assume $\Omega$ is a topological space and endowed with the Borel $\sigma$-algebra. 
The Dirac measure at a point $\omega \in \Omega$, i.e., the unit point mass at $\omega$, is denoted by $\delta_\omega$. 
A distribution $p \in \calD(\bbR)$ has \emph{finite support} if there is a finite set $D \subset \bbR$ such that $p(D)=1$; or has \emph{bounded support} if there is $x \in \bbR_{> 0}$ such that $p([-x,x])=1$. 

For a formula $\varphi(\omega)$ whose assignment $\omega$ is in a set $\Omega$, we define a set $\sem{\varphi}$ by $\sem{\varphi} = \{\omega\in \Omega \mid \varphi(\omega) \mbox{ is true}\}$. For a probability space $(\Omega,\calF,\bbP)$ and a formula $\varphi$ with its assignment set $\Omega$, we call the value $\bbP(\sem{\varphi})$ the \emph{satisfaction probability} of $\varphi$, provided $\sem{\varphi} \in \calF$; we write $\bbP(\varphi)$ to denote this value. 
We say \emph{$\varphi$ holds $\bbP$-almost surely} ($\bbP$-a.s.) if $\bbP(\varphi) = 1$. 
We may omit the ``$\bbP$-'' prefix when the underlying probability space is clear from the context.
\subsection{Markov Decision Processes and Program MDPs}
We consider two formalisms as a model of probabilistic programs: 
One is \emph{Markov Decision Processes} (MDPs), a standard formalism to describe discrete-time probabilistic systems with non-determinism. 
Another is \emph{program MDPs}, a particular subclass of MDPs we introduce in this paper. 
As mentioned in Section~\ref{sec:introduction}, program MDPs are designed as an abstraction of MDPs induced from \emph{probabilistic control flow graphs}, which is another standard semantics of probabilistic programs. 
Our aim is to make it easier to trace where the characteristic properties of
such induced MDPs are used in our technical development, in particular in the
soundness proof of lazy SSMs.

A \emph{Markov Decision Process} (MDP) is a tuple $\calM = (S,A,P, S_I)$, where $S$ is a non-empty set of states, $A$ is a non-empty set of actions, $P:S\times A \to \calD(S)$ is a transition kernel, and $S_I \subseteq S$ is a set of initial states.
The sets $S$ and $A$ are equipped with some $\sigma$-algebras, and $P$ and $S_I$ are required to satisfy a standard measurability condition, see e.g.~\cite{BertsekasS07}. 
An MDP is called a \emph{Markov Chain} when $A$ is a singleton, in which case, we omit the description of $A$ and identify $P$ with a function of the type $P: S \to \calD(S)$. 

Below we give our definition of program MDPs. Recall we say a distribution is Dirac when it is the unit point mass at some point.

\begin{mydefinition}
    A \emph{program MDP} is an MDP $\calM = (S,A,P, S_I)$ that satisfies the following:
    \begin{itemize}
        \item (State set characterization) The state set $S$ is of the form $S = L \times \bbR^n$, where $L$ is a finite set of locations and $\bbR^n$ is the set of assignment vectors, for some $n\in \bbN$.
        \item (Independent updates) The transition kernel $P$ is decomposed into $P_L:S \times A \to \calD(L)$ and $P_{v_1}, \ldots, P_{v_n}:S \times A \to \calD(\bbR)$, i.e., we have the following for each $s\in S$ and $a\in A$:
        \[
        P(s,a) = P_L(s,a)\times P_{v_1}(s,a)\times \ldots\times P_{v_n}(s,a).
        \]
        \item (One-by-one randomization) 
        For each $s\in S$ and $a \in A$, only up to one distribution in $P_L(s,a), P_{v_1}(s,a), \ldots, P_{v_n}(s,a)$ is non-Dirac.
        \item (Finite variety of randomization) 
        There is a finite set $\{p_1, \ldots, p_{k} \}\subseteq \calD(L)$ such that,
        for any $(s,a) \in S \times A$, the distribution $P_L(s,a)$ is either Dirac or $P_L(s,a) = p_j$ for some $j \in \{1, \ldots, k\}$. 
        Similarly, there exists a finite set $\{q_1, \ldots, q_{k'}\} \subseteq \calD(\bbR)$ such that,
        for any $(s,a) \in S \times A$ and $1 \leq i \leq n$, 
        the distribution $P_{v_i}(s,a)$ is Dirac or $P_{v_i}(s,a) = q_j$ for some $j \in \{1, \ldots, k'\}$. 
        We call $\{p_1, \ldots, p_{k} \}$ and $\{q_1, \ldots, q_{k'}\}$ the \emph{constructs} of $P$; we call $\{q_1, \ldots, q_{k'}\}$ the \emph{variable construct} of $P$.
    \end{itemize}
\end{mydefinition}
Below we explain some ideas of the conditions above. \emph{Independent updates} says 
that, by taking an action $a$ at a state $s$, the successor state $(\ell',\myvec{x}')$ is determined by the following sampling procedure:
 \[
\mbox{``} \ell' \sim P_L(s,a) 
 \quad,\quad
 x_1' \sim P_{v_1}(s,a)
 \quad,\quad
 \ldots
 \quad,\quad
 x_n' \sim P_{v_n}(s,a)\mbox{''}.
 \]
\emph{One-by-one randomization} says that 
only up to one location or variable is updated w.r.t. a non-Dirac distribution per transition, and others are updated non-probabilistically, e.g., ``$\ell' \leftarrow f_L(s,a)$'' or ``$x_i' \leftarrow f_{v_i}(s,a)$'' for some functions $f_L:S\times A \to L$ and $f_{v_i}: S\times A \to \bbR$. 
Finally, \emph{finite variety of randomization} says that non-Dirac updates of location and variables must be done according to the constructs of $P$, e.g., ``$\ell' \sim p_j$'' or  ``$x_i' \sim q_j$'', where the constructs of $P$ are finite. 
Notice that we do not impose such a restriction on Dirac updates: 
For example, %
a self-loop at $(\ell, x)$ is realized as $P((\ell,x),a) = \delta_\ell \times \delta_{x}$, and $P$ is allowed to have such a behavior %
for every $x \in \bbR$.

Program MDP is general enough as a semantics of probabilistic programs in e.g.~\cite{TakisakaZWL24}: 
It can realize the standard control structure in imperative languages such as 
if-branches and while-loops, as well as non-deterministic branching $\mbox{`\textbf{if}} \, \mbox{$\star$'}$, non-deterministic
assignment $ \mbox{`\textbf{$x\leftarrow$ndet$(D)$}'}$ with  $D \subseteq \bbR$, 
probabilistic branching $\mbox{`\textbf{if}} \, \mbox{\textbf{prob$(p)$}'}$ with 
$p \in [0,1]$, and 
probabilistic assignment $\mbox{`\textbf{$x \sim \mu$}'}$ with $\mu \in \calD(\bbR)$.

\begin{myremark}
One caveat is that we mean $p$ and $\mu$ above are constants, i.e., their descriptions do not involve program variables: 
Due to \emph{finite variety of randomization}, program MDPs cannot realize $\mbox{`\textbf{if}} \, \mbox{\textbf{prob$(x)$}'}$ by a single MDP transition when $x$ is a program variable. 
Such a restriction in probabilistic program is common, however~\cite{TakisakaZWL24}; 
one can also mimic such a behavior by combining a few program lines. For example, one can mimic 
$\mbox{`\textbf{if}} \, \mbox{\textbf{prob$(x)$}'}$ by 
$\mbox{`$y\sim \mathrm{Unif}[0,1]; \ $\textbf{if}} \, \mbox{$y\leq x$'}$.
\end{myremark}

We recall standard relevant notions about MDPs. 
For a given MDP $\calM = (S,A,P,S_I)$, a \emph{run} of $\calM$ is an infinite sequence $s_0s_1\ldots$ of states such that $s_0 \in S_I$; 
a \emph{finite path} is a finite, nonempty prefix of a run. 
The set of all runs and finite paths of $\calM$ are denoted by $\Pi_\calM$ and $\Pi_\calM^\mathsf{fin}$, respectively. 
An \emph{invariant} is a measurable set $\inv \subseteq S$ that subsumes $S_I$ and is (almost surely) closed under transitions, i.e., $S_I \subseteq \inv$ and $\mynext_a \mychar_{\inv}(s) = 1$ for each $s \in \inv$ and $a\in A$. 

A \emph{scheduler}\footnote{The notation $\sigma$ is used for schedulers and the terminology ``$\sigma$-algebra'', which are separate notions.} of $\calM$ is a function $\sigma: \Pi_\calM^\mathsf{fin} \to \calD(A)$. 
We say $\sigma$ is \emph{deterministic} when $\sigma(w)$ is Dirac for every $w \in \Pi_\calM^\mathsf{fin}$, in which case, we may identify $\sigma$ with a function of the type $\sigma: \Pi_\calM^\mathsf{fin} \to A$ (i.e., if $\sigma(w) = \delta_a$ then we may write $\sigma(w) =a$ instead).

For an initial state $s_I \in S_I$ and a scheduler $\sigma$ of $\calM$, 
there is a standard way to construct a probability space $(\Pi_\calM, \calF, \bbP_{s_I}^\sigma)$ such that, for any measurable set $W \subseteq \Pi_\calM$ of runs of $\calM$, the value $\bbP_{s_I}^\sigma(W)$ is the probability that an execution of $\calM$ from $s_I$ under $\sigma$ generates a run $w \in W$. 
We call such a probability space the \emph{dynamics of $\calM$ under $\sigma$ and $s_I$}. 
See e.g.~\cite{BertsekasS07} for its formal construction.

A \emph{path property} of $\calM$ is a measurable set $\Omega \subseteq \Pi_\calM$. 
We say $\calM$ satisfies the property $\Omega$ almost surely under $s_I$ and $\sigma$ if $\bbP_{s_I}^\sigma(\Omega) = 1$; if $\calM$ satisfies $\Omega$ almost surely for any $s_I\in S_I$ and deterministic scheduler $\sigma$ (in which case, it does so for any $\sigma$ as well~\cite{Feinberg96}), we simply say $\calM$ satisfies the property $\Omega$ almost surely.

\subsection{Preliminaries of Martingales}%
We recall two standard base formalisms to introduce martingale notions~\cite{chakarov2013probabilistic,ChatterjeeGNZZ23}.  
One is a real-valued function $\eta: S \to \bbR$ over states $S$ of an MDP $\calM$, often called \emph{measurable maps}. 
Another is a  
\emph{stochastic process}, i.e., a sequence $\seqOf{X_t}$ of real-valued random variables $X_t$. 
The two formalisms are closely related to each other: 
For a given initial state and a scheduler of $\calM$, 
the behavior of the value of an MM $\eta$ through an execution of $\calM$ forms a stochastic process, which we call an \emph{induced process}.

\paragraph{Measurable maps and pre-expectation}
Fix an MDP $\calM = (S,A,P, S_I)$. %
A measurable map (MM) is a Borel measurable function %
$\eta:S \to \bbR$. %
For a given MM $\eta$, $s \in S$ and $a \in A$,  
the \emph{pre-expectation} of $\eta$ under $a$ is defined by $\bbX_a \eta(s) = \int_S \eta \ dP(s,a)$: that is, the expected value of $\eta$ after a transition from $s$ under $a$. We write $\bbX \eta(s)$ instead when $\calM$ is a Markov chain. 
The \emph{maximal} pre-expectation of $\eta$ is also defined\footnote{We note that $\overline{\bbX}\eta$ may not be measurable even if $\eta$ is, as discussed in~\cite{TakisakaOUH21}. This does not incur technical problems in this paper, however, because we use this function only as a convenient notation to describe global properties over the action set $A$. For example, $\overline{\bbX}\eta (s)\leq \eta(s)$ is equivalent to $\forall a \in A. \bbX_a\eta (s)\leq \eta(s)$ as a formula.} by 
$\overline{\bbX} \eta(s) = \sup_{a \in A} \bbX_a\eta(s)$. 
Also, for a given scheduler $\sigma$, the pre-expectation of $\eta$ under $\sigma$ is defined as follows: 
for $w =s_0\cdots s_t\in \Pi_\calM^\mathsf{fin}$, %
we let  
$\bbX_\sigma \eta (w) = \int_{a \in A} \bbX_a \eta(s_t) \ d\sigma(w)$. 
In particular, when $\sigma$ is deterministic and $\sigma(w) =  \delta_a$, then $\bbX_\sigma\eta(w) = \bbX_a\eta(s_t)$.
Similar to $\bbX_a\eta$, the value $\bbX_\sigma\eta(w)$ represents the expected value of $\eta$ after a transition from the current transition history $w$, following the scheduler $\sigma$.

\paragraph{Stochastic processes and conditional expectation} 
A (discrete-time) \emph{stochastic process} in a probability space $(\Omega, \calF, \bbP)$ is a sequence $(X_t)_{t=0}^\infty$ of %
$\calF$-measurable random variables $X_t: \Omega \to \bbR$ for $t \in \bbN$. 
As standard technical requirements, we assume $(\Omega, \calF, \bbP)$ is associated with a \emph{filtration} $(\calF_t)_{t=0}^\infty$, i.e., a sequence of sub-$\sigma$-algebras of $\calF$ such that $\calF_t \subseteq \calF_{t+1}$ for each $t \in \bbN$; 
and 
$(X_t)_{t=0}^\infty$ is \emph{adapted} to $(\calF_t)_{t=0}^\infty$, i.e., $X_t$ is $\calF_t$-measurable for each $t\in \bbN$.

For a probability space $(\Omega, \calF, \bbP)$, an $\calF$-measurable random variable $X: \Omega \to \bbR$ and a sub-$\sigma$-algebra $\calF'$ of $\calF$, 
a \emph{conditional expectation} of $X$ w.r.t. $\calF'$ is a function $\bbE[X \mid \calF']: \Omega \to \overline{\bbR}$ 
such that $\int_A\bbE[X \mid \calF'] \ d\bbP = \int_A X \ d\bbP$ for each $A \in \calF'$.
A conditional expectation uniquely exists (up to the difference over a null set) when $\bbE_\bbP[X]$ exists, i.e., when either $\bbE_\bbP[X^+]<\infty$ or $\bbE_\bbP[X^-]<\infty$ hold.

For a stochastic process $\seqOf{X_t}$, the conditional expectation $\bbE[X_{t+1} \mid \calF_t]$ plays the role of pre-expectation for MMs. 
Intuitively, $\bbE[X_{t+1} \mid \calF_t](\omega)$ is the expected value of $X_{t+1}$ given the information about $\omega$ at time $t$; 
similarly, $\bbX_\sigma\eta(w)$ can be seen as the expected value of $\eta$ at the next time frame, given the information $w$ about the current execution history of $\calM$.

\paragraph{Induced stochastic processes from measurable maps} 
Let an MDP $\calM = (S,A,P, S_I)$ and MM $\eta: S \to \bbR$ be given.
For a scheduler $\sigma$ and an initial state $s_I$ of $\calM$, we call the following $\seqOf{X_t}$ an \emph{induced process} from $\eta$ under $s_I$ and $\sigma$: 
For $\omega = s_0s_1\cdots$ and $t\in \bbN$, we let $X_t(\omega) = \eta(s_t)$. 
Observe it is a stochastic process over $(\Pi_{\calM}, \calF, \bbP_{s_I}^\sigma)$, i.e., the dynamics of $\calM$ under $\sigma$ and $s_I$. There is also a canonical filtration $\seqOf{\calF_t}$ associated to $(\Pi_{\calM}, \calF, \bbP_{s_I}^\sigma)$, generated by so-called \emph{cylinder set construction}, and $\seqOf{X_t}$ is adapted to it.

For induced processes, the aforementioned correspondence between an MM $\eta$ and its induced process $\seqOf{X_t}$ is formally justified as follows, provided that $\bbE[X_t]$ exists for each $t \in \bbN$.

\begin{myproposition}[e.g.~\cite{TakisakaZWL24}]\label{prop:condExpConventional}
    Let an MDP $\calM$ and an MM $\eta$ over $\calM$ be given, and let $\seqOf{X_t}$ be the induced process from $\eta$ under a scheduler $\sigma$ and an initial state $s_I$. 
    For $t \in \bbN$,
    if $\bbE[X_{t+1}]$ exists, %
    then the function $f_t(s_0s_1\ldots) = \bbX_\sigma \eta(s_0\ldots s_t)$ is a conditional expectation $\bbE[X_{t+1} | \calF_t]$. \qed
\end{myproposition}

\section{Relatively Well-Behaved Distributions and Their Coverage}\label{sec:RelativeWellBehaved}

In this section, we introduce the key notion explained in
Section~\ref{subsec:polySoundness}: \emph{relatively well-behaved
distributions}. We also explain why the original well-behavedness assumption yields an overly narrow soundness
theorem in the polynomial-template setting. We proceed as follows:
\begin{itemize}
    \item As a recap, we recall the original definition of well-behavedness
    from~\cite{TakisakaZWL24} (Definition~\ref{def:wellBehavedOriginal}).
    This notion is tailored to the soundness proof of linear lazy RSMs.
    We also provide an overview of the soundness proof of linear lazy RSMs
    given in~\cite{TakisakaZWL24}, in order to recall how the well-behavedness
    condition is used there.

    \item We introduce the notion of \emph{well-behavedness over $\calG$}
    (Definition~\ref{def:wellBehavednessG}), where $\calG$ is a user-defined
    function class. We observe that this notion coincides with
    Definition~\ref{def:wellBehavedOriginal} when $\calG$ is the class of
    linear functions (Proposition~\ref{prop:wellBehavedCoincidence}), thereby
    justifying our definition. We then show that uniform distributions over
    intervals are not well-behaved over polynomial functions
    (Proposition~\ref{prop:CounterExWellBehavedness}). Thus, while this generalized
    well-behavedness condition would still support a soundness proof, the
    resulting theorem would have very limited practical coverage for polynomial
    lazy RSMs.

    \item Finally, we introduce the notion of relative well-behavedness over
    $\calG$ (Definition~\ref{def:weakWellBehaved}), and show that every
    bounded-support distribution is relatively well-behaved over polynomial
    functions (Proposition~\ref{prop:WeakWellBehaved}). This gives a
    substantially broader and practically meaningful scope for our soundness
    theorem.
\end{itemize}

\subsection{A Recap: Well-Behaved Distributions and Their Use in the Existing Proof}
Below, for $a \in \bbR$, let $\calF_a$ be a set of functions defined by $\calF_a = \{f(x) = ax+b \mid b\in \bbR\}$: That is, $\calF_a$ is the set of all univariate linear functions with the angle $a$.

\begin{mydefinition}[well-behaved distributions~\cite{TakisakaZWL24}]\label{def:wellBehavedOriginal}
    We say a distribution $p \in \calD(\bbR)$ is \emph{well-behaved} if the following holds: for any $a \in \bbR $, there exist constants $C_1 \in (0,1)$ and $C_2 >0$ such that%
\begin{align}
    \forall f\in \calF_a. 
    \biggl[
        p(f < 0) \leq C_1 \Rightarrow 
        \int f^- dp \leq C_2 \cdot p(f < 0)
    \biggr]. \label{eq:wellBehavednessOriginal}
\end{align}
\end{mydefinition}
This definition formalizes an informal explanation in Section~\ref{subsec:polySoundness} as follows.
The argument there was, in a general form, to consider any linear function $g:\bbR^n \to \bbR$, $i\in \{1, \ldots, n\}$, $\myvec{x}_{-i} \in \bbR^{n-1}$, and the random variable
\begin{align}
    g_{\myvec{x}_{-i}}(x_i') = g(x_i',\myvec{x}_{-i}), \qquad \mbox{where} \qquad x_i' \sim p.
\end{align}
Then we consider the relationship of $p(g_{\myvec{x}_{-i}}<0)$ and $\int(g_{\myvec{x}_{-i}})^- \ dp$ when $\myvec{x}_{-i} \in \bbR^{n-1}$ can be arbitrary. 
Condition (\ref{eq:wellBehavednessOriginal}) implies what corresponds to (\ref{eq:convergenceRel}) in Section~\ref{subsec:polySoundness}: for any sequence
$(f_n)_{n\in \bbN} \subseteq \calF_a$ such that %
$p(f_n<0) \rightarrow 0$,
condition (\ref{eq:wellBehavednessOriginal}) implies 
$\int (f_n)^-\,dp\to 0$.
In fact, more strongly, 
\begin{align}
    \int (f_n)^-dp = O\bigl(p(f_n<0)\bigr) \qquad \mbox{as \, $n \to \infty$}.\label{eq:linearVanish}
\end{align}
That is, $\int (f_n)^-dp$ vanishes at a linear or faster rate relative to $p(f_n<0)$. 

\begin{myexample}
    Let $p = \mathrm{Unif}[0,1]$, the uniform distribution over the unit interval. Then $p$ is well-behaved. To see this, for any $a \in \bbR \setminus \{0\}$, let $C_1$ be any value in $(0,1)$ and $C_2 =|a|$. Then the condition (\ref{eq:wellBehavednessOriginal}) holds under such $C_1$ and $C_2$. 
    Indeed, for any $f\in \calF_a$ such that $p(f<0)<1$, we have $\max_{0\leq x \leq 1}f(x)>0$ (as $p(f(x)>0)$ is non-zero), and hence $\min_{0\leq x \leq 1}f(x) > -|a|$. Thus we have
    \[
     \int f^- dp = -\int_{\sem{f<0}}f \ dp \leq -\min_{0\leq x \leq 1}f(x) \cdot p(f < 0) < |a| \cdot p(f < 0).
    \]
    For $a=0$, condition (\ref{eq:wellBehavednessOriginal}) trivially holds under any $C_1 \in (0,1)$ and $C_2 >0$ because any $f \in \calF_a$ is a constant function, and thus we have $p(f<0) \leq C_1 \implies p(f<0) =0 \implies  \int f^- dp = 0$. 
    
    By a similar argument, any $p\in \calD(\bbR)$ with bounded support is shown to be well-behaved. Unbounded-support distributions can also be well-behaved when their tail probabilities vanish quickly enough: For example, normal distributions are well-behaved~\cite[Proposition C.6]{TakisakaZWL24arXiv}. \demo
\end{myexample}

Now we give a sketch of the soundness proof of a linear lazy LexRSM in~\cite{TakisakaZWL24} to see how the well-behavedness condition is involved in it. 
We give a simplified version in which we focus on a one-dimensional lazy lexicographic RSM (which we call a lazy RSM for brevity) and omit nondeterminism. %
Under such a simplification, the notion is described as follows.

\begin{mydefinition}[lazy RSM~\cite{TakisakaZWL24}]
    Let $\calM = (S, P, S_I)$ be a program Markov chain (i.e., a program MDP such that $A$ is a singleton) with a terminal location $\ell_\mathsf{term} \in L$, and let $T = \{(\ell_\mathsf{term},\myvec{x})\mid \myvec{x} \in \bbR^n\} \subseteq S$. An MM $\eta$ over $\calM$ is called a \emph{lazy ranking supermartingale} (lazy RSM) if it satisfies the following for some invariant $\inv \subseteq S$:
    \begin{align*}
    \mbox{(Ranking condition)} \quad
      & \forall s \in \inv\setminus T. \ 
        [\bbX\eta(s) \leq \eta(s)-1],\\      
    \mbox{(Weak non-negativity)} \quad 
      & \forall s \in \inv\setminus T. \ 
        [\eta(s) \geq 0].
\end{align*}
\end{mydefinition}
The only difference from the standard RSM~\cite{chakarov2013probabilistic} is that the non-negativity condition is not required in the terminal region $T$. 
It has been shown in~\cite{TakisakaZWL24} that a lazy RSM $\eta$ is sound, i.e., the underlying $\calM$ almost surely reaches $T$ eventually, if the following hold: 
\begin{enumerate}[(i)]
    \item $\eta$ is linear, i.e., $\eta_\ell(\myvec{x}) = \eta(\ell, \myvec{x})$ is a linear function for each $\ell \in L$,\label{item:LazyRSMSoundness1}
    \item $\calM$ is induced from a linear arithmetic PP, and \label{item:LazyRSMSoundness2}
    \item $\calM$ is induced from a PP whose variable samplings follow 
    well-behaved distributions. %
    In the program MDP formalism, it amounts to requiring the distributions in the variable construct of $P$ to be well-behaved. \label{item:LazyRSMSoundness3}%
\end{enumerate}
We found the condition (\ref{item:LazyRSMSoundness2}) can be simply dropped: We explain in Remark~\ref{rem:WhyArithmeticUnnecessary} why the original proof is concerned about it, and how we find it redundant. %
The conditions (\ref{item:LazyRSMSoundness1}) and (\ref{item:LazyRSMSoundness3}), on the other hand, 
are crucial. 
The proof plan is to show that the following function $\tilde\eta$, called the \emph{$\varepsilon$-fixing} of $\eta$, retains the ranking condition at particular states for an appropriate $\varepsilon>0$. Observe the value of $\tilde\eta$ has a uniform lower bound $-\varepsilon$, so we do not have an issue of unbounded negativity with it.
\begin{align}
    \Tilde{\eta}(s)=
    \begin{cases}
        \eta(s) & \mbox{if $\eta(s)\geq 0$,}\\
        -\varepsilon & \mbox{otherwise.}
    \end{cases}
\end{align}%
To show that, we observe the condition (\ref{eq:wellBehavednessOriginal}) in well-behavedness can be rewritten as follows:
\begin{align}
    \forall f\in \calF_a. 
    \biggl[
        p(f < 0) \leq C_1 \Rightarrow 
        \int f dp \geq \int f^+ dp - C_2 \cdot p(f < 0)
    \biggr].%
\end{align}
This ensures the existence of $\varepsilon >0$ and $\gamma \in (0,1)$ such that, for each program location $\ell \neq \ell_\mathsf{term}$ with the associated program line `$x \sim p$', %
a linear lazy RSM $\eta$ and its $\varepsilon$-rescaling $\tilde{\eta}$ satisfy
\begin{align}
    \forall \myvec{x}. 
    \bigl[
    \pneg(\ell, \myvec{x}) \leq \gamma \implies \bbX \eta(\ell, \myvec{x}) \geq \bbX\Tilde{\eta}(\ell, \myvec{x}) 
    \bigr],
    \label{eq:translateWellBehavedness}
\end{align}
where $\pneg(\ell, \myvec{x})$ represents the probability that the value of $\eta$ is negative after a transition from $(\ell, \myvec{x})$ (or, more formally, $\pneg(\ell, \myvec{x}) = \bbX \mychar_{\sem{\eta <0}}(\ell, \myvec{x})$). 
We also have $\bbX \eta (\ell, \myvec{x}) \leq \eta(\ell, \myvec{x})-1$ and $\eta(\ell, \myvec{x}) \geq 0$ by the ranking and non-negativity conditions; the latter further implies $\tilde\eta(\ell, \myvec{x}) = \eta(\ell, \myvec{x})$. %
Overall, we have the following, which claims $\tilde\eta$ retains the ranking condition at $\ell$ given $\pneg(\ell, \myvec{x}) \leq \gamma$:
\begin{align}
    \forall \myvec{x}. 
    \bigl[
    \pneg(\ell, \myvec{x}) \leq \gamma \implies \bbX\Tilde{\eta}(\ell, \myvec{x}) \leq \tilde\eta(\ell, \myvec{x}) -1
    \bigr].
    \label{eq:translateWellBehavedness2}
\end{align}
It can be shown that (\ref{eq:translateWellBehavedness2}) also holds at any other non-terminal location (by a much simpler argument). 
This property---called \emph{$(\varepsilon, \gamma)$-rescalability} of $\eta$ in~\cite{TakisakaZWL24}---ensures that the underlying program is AST: 
Although $\tilde{\eta}$ satisfies the ranking condition only when $\pneg(\ell, \myvec{x}) \leq \gamma$, 
such $\tilde{\eta}$ witnesses that a program run almost surely satisfies either 
\begin{inparaenum}[(a)]
    \item termination, or 
    \item recurrence to the state set $S_\gamma = \{s \in S \setminus T \mid \pneg(s) > \gamma \}$.
\end{inparaenum}
The latter almost surely implies that the value of $\tilde\eta$ eventually becomes negative during a program run, and hence the program terminates.

\begin{myremark}
    We provide some observations about %
    why well-behavedness is designed to satisfy (\ref{eq:linearVanish}). %
    If we drop the requirement on the convergence speed, then the definition of well-behavedness will be as follows, via the standard epsilon-delta argument:
    \begin{align}
    \forall a\in\bbR. \ 
    \forall C_2 >0. \
    \exists C_1 \in (0,1). \ 
    \forall f\in \calF_a.
    \biggl[
        p(f < 0) \leq C_1 \Rightarrow 
        \int f^- dp \leq C_2
    \biggr]. \label{eq:wellBehavednessNoSpeed}
\end{align}
This is actually enough for soundness \emph{if the ranking condition is imposed at every non-terminal state}. %
By using (\ref{eq:wellBehavednessNoSpeed}) instead, one can obtain the following, which is a slightly weak variant of (\ref{eq:translateWellBehavedness2}):
\begin{align}
    \forall \myvec{x}. 
    \bigl[
    \pneg(\ell, \myvec{x}) \leq \gamma \implies \bbX\Tilde{\eta}(\ell, \myvec{x}) \leq \tilde\eta(\ell, \myvec{x}) -(1-C_2)
    \bigr].
\end{align}
By taking sufficiently small $C_2$, this is enough to claim $\tilde\eta$ is $(\varepsilon,\gamma)$-rescalable. However, in a lazy LexRSM, the ranking condition may be replaced by the \emph{unaffecting condition} at some states, which only requires $\bbX\eta(s) \leq \eta(s)$. The condition (\ref{eq:wellBehavednessNoSpeed}) seems too weak to show $\tilde\eta$ retains this condition: The difference is that the ranking condition has a ``buffer to absorb $C_2$'' while the unaffecting condition does not. In other words, $-1$ can still be negative after adding a small positive number, but $0$ cannot be zero after adding any positive number. \demo
\end{myremark}

\subsection{Well-Behaved Distributions over $\calG$ and Their Limitation}

The well-behavedness in Definition~\ref{def:wellBehavedOriginal}, in particular how we define the function class $\calF_a$, is tailored for the soundness proof of \emph{linear} lazy RSMs. 
Based on how the set $\calF_a$ appeared in the proof argument, the following extension appears to be natural.%

\begin{mydefinition}[well-behaved distributions over $\calG$]\label{def:wellBehavednessG}
    Let $\calG$ be a set of functions such that each $g \in \calG$ is of the type $g:\bbR^k \rightarrow \bbR$ for some $k \in \bbN$ ($\calG$ may contain functions with different $k$'s). 
    We say $p \in \calD(\bbR)$ is \emph{well-behaved over $\calG$} if, for each %
    $g:\bbR^k \to \bbR \in \calG$ 
    and $i \in \{1, \ldots, k\}$, there exist $C_1 \in (0,1)$ and $C_2 >0$ such that
    \begin{align}
    \forall \boldsymbol{x}_{-i}\in \bbR^{k-1}. 
    \biggl[
        p(g_{\boldsymbol{x}_{-i}} < 0) \leq C_1 \Rightarrow%
        \int (g_{\boldsymbol{x}_{-i}})^- dp
        \leq
        C_2 \cdot p(g_{\boldsymbol{x}_{-i}} < 0)
    \biggr],\label{eq:wellBehavednessG}
\end{align}
where $g_{\boldsymbol{x}_{-i}}: \bbR \to \bbR$ is a function that is induced by fixing the value of all variables but $x_i$ by $\boldsymbol{x}_{-i} \in \bbR^{k-1}$, i.e., $g_{\boldsymbol{x}_{-i}}(x_i) = g(x_i, \boldsymbol{x}_{-i})$. 
\end{mydefinition}
In this definition, the set $\calG$ represents the template class we consider, 
and a given $g:\bbR^k \to \bbR \in \calG$ 
    and $i \in \{1, \ldots, k\}$ induce a set $\calF_{g,i} := \{g_{\boldsymbol{x}_{-i}} \mid \boldsymbol{x}_{-i}\in \bbR^{k-1}\}$, which is a generalization of $\calF_a$. Indeed, if we let $\calG$ be the set of all linear functions, then 
    for $g\in \calG$ such that $g(\myvec{x}) = \myvec{a}\cdot\myvec{x} + b$ we have $\calF_{g,i} = \calF_{a_i}$. 
    Thus it is easy to show the following, which justifies Definition~\ref{def:wellBehavednessG} as a generalization of Definition~\ref{def:wellBehavedOriginal}.
\begin{myproposition}\label{prop:wellBehavedCoincidence}
    Let $\calG$ be the set of all linear functions. Then $p \in \calD(\bbR)$ is well-behaved over $\calG$ if and only if $p$ is well-behaved in the sense of Definition~\ref{def:wellBehavedOriginal}. \qed
\end{myproposition}

Based on Definition~\ref{def:wellBehavednessG}, the correctness proof of linear lazy RSMs is canonically modified into one for polynomial lazy RSMs, under the assumption that each distribution in the underlying PP is well-behaved over the set of all polynomial functions. 
The issue is that such an assumption is now too restrictive to cover the class of polynomial PPs: 
When $\calG$ is the set of all polynomial---or even quadratic---functions, 
even uniform distributions over intervals are not well-behaved over $\calG$.

\begin{myproposition}\label{prop:CounterExWellBehavedness}
    Let $\calG$ be the set of all quadratic functions, and let $p \in \calD(\bbR)$ be the uniform distribution over any interval, say $[a,b]$ with $a <b$. Then $p$ is not well-behaved over $\calG$.
\end{myproposition}

\begin{proof}
    We assume $p$ is the uniform distribution over $[0,1]$, the argument for general $a$ and $b$ is similar. 
    Let $g(x_1,x_2,x_3) = x_1 x_2+x_3$, then $g \in \calG$. 
    Also let $f_t(x) = 4^tx-2^{t}$ for each $t = 1,2,\ldots$; then we have $f_t(x) = g(x,4^t, -2^t)$ and thus $f_t \in \calF_{g,1}$. 
    Finally, observe we have $p(f_t<0) \rightarrow 0$ as $t \rightarrow \infty$ while $\int (f_t)^- dp = \frac{1}{2}$ for every $t \geq 1$. Thus, there is no $C_1 \in (0,1)$ and $C_2>0$ that satisfies (\ref{eq:wellBehavednessG}) for such $g$ and $i=1$, meaning that $p$ is not well-behaved over $\calG$.
\end{proof}

\subsection{Relatively Well-Behaved Distributions and Their Coverage}

\begin{mydefinition}\label{def:weakWellBehaved}
  Let $\calG$ be a set of functions such that each $g \in \calG$ is of the type $g:\bbR^k \rightarrow \bbR$ for some $k \in \bbN$ ($\calG$ may contain functions with different $k$'s). 
    We say $p \in \calD(\bbR)$ is \emph{relatively well-behaved over $\calG$} if, for each %
    $g:\bbR^k \to \bbR \in \calG$ 
    and $i \in \{1, \ldots, k\}$, there exist $C_1 \in (0,1)$ and $C_2 >0$ such that
    \begin{align}
    \forall \boldsymbol{x}_{-i}\in \bbR^{k-1}. 
    \biggl[
        p(g_{\boldsymbol{x}_{-i}} < 0) \leq C_1 \Rightarrow 
        \int (g_{\boldsymbol{x}_{-i}})^- dp
        \leq
        C_2 \cdot p(g_{\boldsymbol{x}_{-i}} < 0) \cdot \int g_{\boldsymbol{x}_{-i}} \ dp
    \biggr],\label{eq:weakWellBehavedness}
\end{align}
where $g_{\boldsymbol{x}_{-i}}: \bbR \to \bbR$ is a function that is induced by fixing the value of all variables but $x_i$ by $\boldsymbol{x}_{-i} \in \bbR^{k-1}$, i.e., $g_{\boldsymbol{x}_{-i}}(x_i) = g(x_i, \boldsymbol{x}_{-i})$.  
\end{mydefinition}
The only difference between Definition~\ref{def:wellBehavednessG} 
and 
Definition~\ref{def:weakWellBehaved}
is the consequent part of (\ref{eq:wellBehavednessG}) and (\ref{eq:weakWellBehavedness}). 
Also observe (\ref{eq:weakWellBehavedness}) represents the trade-off relationship that we would like to realize: 
If $p(g_{\boldsymbol{x}_{-i}} < 0) \rightarrow 0$ while $\int (g_{\boldsymbol{x}_{-i}})^- dp$ remains a positive constant, then we have $\int g_{\boldsymbol{x}_{-i}} \ dp \rightarrow \infty$. 
Now we have our main result in this section as follows.

\begin{myproposition}\label{prop:WeakWellBehaved}
    Let $\calG$ be the set of all polynomial functions. Then $p\in \calD(\bbR)$ is relatively well-behaved over $\calG$ if $p$ has bounded support.
\end{myproposition}

We present some technical setup before proving Proposition~\ref{prop:WeakWellBehaved}. 
Let $p \in \calD(\bbR)$ be a measure whose support is bounded and infinite: that is, there exists an interval $I$ such that $p(I) = 1$, but no finite set $J$ such that $p(J) = 1$. Fix an interval $I \subseteq \bbR$ such that $p(I) = 1$. 
For $f:\bbR \to \bbR$, we let
\begin{align}
    \|f\|_{\infty, I}:= \sup_{x \in I} |f(x)|, \qquad 
    \|f\|_{1, p}:= \int |f|dp.
\end{align}

\begin{myproposition}\label{prop:oneNorm}
    Let $p \in \calD(\bbR)$ be a distribution whose support is infinite and bounded. Then for any $k \in \bbN$, the function $\|\cdot\|_{1, p}$ is a norm on $\Poly_k$.
\end{myproposition}

\begin{proof}
    Because a nonzero polynomial has only finitely many zeros, the infinite-support assumption on $p$ implies that $\|\cdot\|_{1, p}$ is positive definite: i.e.,  $\|f\|_{1, p} = 0 \implies f = 0$. Subadditivity and absolute homogeneity are easy to check.
\end{proof}
Two norms $\|\cdot \|_\alpha$ and $\|\cdot \|_\beta$ on a vector space $\calX$ are said to be \emph{equivalent} when there are $C_1, C_2>0$ such that $C_1\|x \|_\alpha \leq \|x \|_\beta\leq C_2 \| x \|_\alpha$ for every $x\in \calX$. The following is a standard result.

\begin{mytheorem}\label{thm:NormUniqueness}
    Any two norms on a finite-dimensional vector space are equivalent. \qed    
\end{mytheorem}

Because $\Poly_k$ is a finite-dimensional vector space with a basis $\{1, x,\ldots, x^k\}$, we know that $\|\cdot\|_{\infty, I}$ and $\|\cdot\|_{1, p}$ are equivalent on $\Poly_k$.

\begin{proof}[Proof of Proposition~\ref{prop:WeakWellBehaved}]
    Take any $p\in \calD(\bbR)$ that has bounded support. 
    We observe that, for any (multivariate) polynomial function $g$ with its degree $k$, we have $\cup_i\calF_{g,i} \subseteq \Poly_k$.
    Therefore, to show $p$ is relatively well-behaved over the set of all polynomial functions, it suffices to prove the following: for any $k \in \bbN$, there exists $C_1\in (0,1)$ and $C_2 >0$ such that 
    \begin{align}
    \forall f \in \Poly_k. 
    \biggl[
        p(f < 0) \leq C_1 \Rightarrow 
        \int f^- dp
        \leq
        C_2 \cdot p(f < 0) \cdot \int f \ dp \label{eq:WeakWellBehavedAlt}
    \biggr].
\end{align}
The claim is vacuously true when $p$ has finite support because in such a case, by taking a sufficiently small $C_1$, we have $p(f < 0) \leq C_1 \Rightarrow p(f < 0) =0$ for any $f\in \Poly_k$. In what follows, suppose $p$ has infinite support.

    Now fix $k \in \bbN$ and take $f \in \Poly_k$, and write $\pneg = p(f<0)$ for brevity. Then we have
    \begin{align}
       \|f^-\|_{1, p} = \int_{\sem{f<0}}|f|dp\leq \pneg \cdot \|f\|_{\infty, I}.
    \end{align}
    Also by Proposition~\ref{prop:oneNorm} and Theorem~\ref{thm:NormUniqueness}, 
    there exists $K>0$ such that $\|f\|_{\infty, I} \leq K\cdot \|f\|_{1, p}$. %
    Thus we have 
    $\|f^-\|_{1, p} \leq \pneg K\cdot \|f\|_{1, p} = \pneg K \cdot (\|f^+\|_{1, p}+\|f^-\|_{1, p})$, 
    and hence 
    $(1-2\pneg K)\cdot \|f^-\|_{1, p} \leq \pneg K\cdot( \|f^+\|_{1, p}- \|f^-\|_{1, p}) = \pneg K\cdot\int f dp$. 
    When $\pneg  \leq \frac{1}{4K}$, this further implies 
    \begin{align}
        \|f^-\|_{1, p} \leq \frac{\pneg K}{1-2\pneg K}\cdot\int f dp \leq 2K \cdot \pneg \cdot \int f dp. 
    \end{align}
    Thus we have (\ref{eq:WeakWellBehavedAlt}) for $C_1 = \frac{1}{4K}$ and $C_2 = 2K$: We note that we can take the same $K$ for any $f \in \Poly_k$.
\end{proof}

\section{Lazy Streett Supermartingale and Its Soundness}\label{sec:LazySSMSoundness}
In this section we introduce our novel supermartingale notion, a \emph{lazy Streett supermartingale} (lazy SSM). We then show it is sound if it is polynomial, and the underlying MDP is a program MDP whose variable construct consists of relatively well-behaved distributions. We proceed as follows:
\begin{itemize}
    \item We give the definition of lazy SSM (Definition~\ref{def:LazySSMMap}), and give the theorem statement of its soundness (Theorem~\ref{thm:SoundnessLSSMmap}).
    \item We give two technical pieces to complete the proof. First, we give the soundness of (strongly non-negative) SSM with two generalizations (Theorem~\ref{thm:SoundnessStrongSSMMap}, \ref{thm:soundnessSSM}): 
    we extend the notion for MDPs, and relax the ranking condition to the \emph{antitone ranking condition}~\cite{Kenyon-RobertsO21}. A high-level plan of our soundness proof is to consider a fractional-power transform of a lazy SSM $\eta$, i.e., a function $(\eta^+)^\theta$, and show it is a strongly non-negative SSM.
    \item Second, we give a key inequality that holds for relatively well-behaved distributions (Lemma~\ref{lem:keyIneq}). 
    It lets us translate the ranking/unaffecting condition of $\eta$ to that of $(\eta^+)^\theta$.
    \item Finally, we give a proof of soundness.
\end{itemize}

\subsection{Lazy Streett Supermartingale: Formal Definition and Soundness}
We first recall several notions that are necessary for our definition. 
For an MDP $\calM = (S,A,P,S_I)$, a \emph{Streett pair} $(\Sfin, \Sinf)$ over $\calM$ is a pair of measurable sets $\Sfin, \Sinf \subseteq S$. 
For technical convenience, we WLOG assume $\Sfin \cap \Sinf = \emptyset$ and write $\Sany = S \setminus (\Sfin \cup \Sinf)$. We say an infinite path $s_0s_1\ldots$ of $\calM$ satisfies the Streett condition $(\Sfin, \Sinf)$ when we have
\begin{align}
    ( \overset{\infty}{\forall}t.s_t\not\in\Sfin) \lor ( \overset{\infty}{\exists}t.s_t\in\Sinf).
\end{align}

We introduce properties of MMs relevant to martingale as follows, where we implicitly fix a non-increasing function $c: \bbR_{\geq 0} \to \bbR_{>0}$. For an MM $\eta$ over %
$\calM = (S,A,P,S_I)$ and $s \in S$, we let:
 \begin{align*}
    \mbox{(Ranking condition)} && \mathsf{Rank}_\eta(s) &\equiv \overline\mynext \eta(s) \leq \eta(s)-c(\eta(s)), \\
    \mbox{(Unaffecting condition)} && \mathsf{UnAffect}_\eta(s) &\equiv \overline\mynext \eta(s) \leq \eta(s), \\
    \mbox{(Non-negativity)} && \mathsf{NonNeg}_\eta(s) &\equiv  \eta(s) \geq 0, \\
    \mbox{(Stability at negativity)} && \mathsf{StabNeg}_\eta(s) &\equiv \eta(s) < 0 \Rightarrow 
    \forall a \in A. \ P(s,a)(\eta<0) = 1.
\end{align*}

A few comments on these notions are in order. 
First, the ranking condition we adopted is the so-called \emph{antitone ranking condition} with respect to $c$~\cite{Kenyon-RobertsO21}. 
This is more general than the typical ranking condition, which requires the value of $\eta(s)$ to decrease by a fixed constant in expectation. 
The antitone ranking condition is strictly weaker than the conventional one when $\lim_{x \to \infty} c(x) =0$. %
We found this relaxation necessary for our proof: As far as we tried, we were only able to transform a lazy SSM into a strongly non-negative SSM with the antitone ranking condition. See also Remark~\ref{rem:whyAntitone}.

Second, \emph{stability at negativity} is a weakening of the non-negativity condition proposed in the recent study of weakly non-negative RSMs~\cite{TakisakaZWL24}. 
The condition says that, if the value of $\eta$ is negative at the current state, then it should remain negative after a transition almost surely. 
It is easy to see $\mathsf{NonNeg}_\eta(s)$ implies $\mathsf{StabNeg}_\eta(s)$.

Now we define our core supermartingale notion, a lazy Streett supermartingale, as follows.

\begin{mydefinition}[lazy Streett supermartingale map]\label{def:LazySSMMap}
Let an MDP $\calM = (S,A,P, S_I)$ be given. An MM $\eta: S \to \bbR$ is called a \emph{lazy Streett supermartingale map (lazy SSM map)} for a Streett pair $(\Sfin, \Sinf)$ if it satisfies the following for some invariant $\inv$ and antitone function $c$:
    \begin{itemize}
        \item We have $\mathsf{Rank}_\eta(s)$ and $\mathsf{NonNeg}_\eta(s)$ for each $s \in \inv \cap \Sfin$.%
        \item We have $\mathsf{UnAffect}_\eta(s)$ and  $\mathsf{StabNeg}_\eta(s)$ for each $s \in \inv \cap \Sany$.
    \end{itemize}
\end{mydefinition}

The definition of lazy SSM itself is a canonical combination of two existing notions, \emph{lazy (Lex)RSM} \cite{TakisakaZWL24} and \emph{(strongly non-negative) SSM}~\cite{KuraU26}. 
As we have discussed already, the main technical novelty lies in its soundness proof, which we now state formally: a lazy SSM $\eta$ is sound if 
\begin{inparaenum}[(a)]
    \item $\eta$ is polynomial, 
    \item the underlying $\calM$ is a program MDP, and
    \item the variable construct of $P$ consists of relatively well-behaved distributions.
\end{inparaenum}

\begin{mytheorem}\label{thm:SoundnessLSSMmap}
    Let $\calM=(S,A,P,S_I)$ be a program MDP such that every distribution in the variable construct of $P$ is relatively well-behaved over polynomial functions. %
    If $\calM$ admits a polynomial lazy SSM map for $(\Sfin, \Sinf)$, then  $\calM$ satisfies the Streett condition $(\Sfin, \Sinf)$ almost surely.  
\end{mytheorem}

\subsection{Soundness of Strongly Non-Negative Streett Supermartingale}
As noted above, we give a few necessary technical pieces before proving Theorem~\ref{thm:SoundnessLSSMmap}. 
Here we prove soundness of a \emph{(strongly non-negative) Streett supermartingale}~\cite{KuraU26} with two updates: 
adaptation to non-determinism, and relaxation of the ranking condition to the antitone one.

\begin{mydefinition}[Streett supermartingale map]\label{def:StrongSSMMap}
   Let an MDP $\calM = (S,A,P,S_I)$ be given. An MM $\eta: S \to \bbR$ is called a \emph{(strongly non-negative) Streett supermartingale map (SSM map)} for a Streett pair $(\Sfin, \Sinf)$ if it satisfies the following for some invariant $\inv$ and antitone function $c$:
    \begin{itemize}
        \item We have $\mathsf{Rank}_\eta(s)$ for each $s \in \inv \cap \Sfin$.%
        \item We have $\mathsf{UnAffect}_\eta(s)$ for each $s \in \inv \cap \Sany$.
        \item  We have $\mathsf{NonNeg}_\eta(s)$ for each $s \in \inv$. 
    \end{itemize}
\end{mydefinition}

The soundness of SSM maps is stated as follows. We note that the result holds over MDPs in general, not only over program MDPs.

\begin{mytheorem}\label{thm:SoundnessStrongSSMMap}
    Suppose an MDP $\calM$ admits a strongly non-negative SSM map $\eta$ for $(\Sfin, \Sinf)$. Then under any scheduler, the run of $\calM$ satisfies the Streett condition $(\Sfin, \Sinf)$ almost surely.
\end{mytheorem}

\begin{myremark}\label{rem:positiveRec}
    It has been shown in~\cite{KuraU26} that, under their definition, an SSM map witnesses \emph{positive almost-sure recurrence}, which is a stronger property than almost-sure satisfaction of Streett conditions. 
    We note that, under our definition (Definition~\ref{def:StrongSSMMap}), an SSM map does \emph{not} witness positive almost-sure recurrence in general, because it adopts the antitone ranking condition. 
    A counterexample is that it can witness almost-sure termination of the symmetric random walk~\cite{Kenyon-RobertsO21}, which is a famous example that is almost surely recurrent but not positively so. \demo
\end{myremark}

Our proof of Theorem~\ref{thm:SoundnessStrongSSMMap} 
follows a standard analysis procedure for \emph{induced stochastic processes}: 
We define a corresponding martingale notion over stochastic processes (Definition~\ref{def:SSM}), show its soundness (Theorem~\ref{thm:soundnessSSM}), and apply it to stochastic processes induced from an SSM map to show its soundness.
\begin{mydefinition}[Streett condition over a measurable space] 
Let $(\Omega,\calF)$ be a measurable space, and $\seqOf{\calF_t}$ be its filtration. A \emph{Streett pair} over $(\Omega,\calF)$ is a tuple $\bigl( \seqOf{\Ofin},\seqOf{\Oinf} \bigr)$ such that $\Ofin, \Oinf \in \calF_t$  for each $t\in \bbN$. For convenience, we assume $\Ofin \cap \Oinf = \emptyset$ for each $t\in \bbN$.%
\end{mydefinition}

\begin{mydefinition}[Streett supermartingale]\label{def:SSM}
Let a probability space $(\Omega, \calF, \bbP)$ be given with its filtration $\seqOf{\calF_t}$. 
We call a stochastic process $\seqOf{X_t}$, adapted to $\seqOf{\calF_t}$, a \emph{Streett supermartingale} for a Streett pair $\bigl( \seqOf{\Ofin},\seqOf{\Oinf} \bigr)$ if it satisfies the following $\bbP$-a.s., for each $t\in \bbN$: below, we write $\Oany = \Omega \setminus (\Ofin \cup \Oinf)$, and $c: [0,\infty) \to (0,\infty)$ is a non-increasing function.
\begin{itemize}
    \item $X_t(\omega) \geq 0$ for each $\omega \in \Omega$.
    \item $\bbE[X_{t+1}\mid \calF_t](\omega) \leq X_t(\omega)-c(X_t(\omega))$ for each $\omega \in \Ofin$.%
    \item $\bbE[X_{t+1}\mid \calF_t](\omega) \leq X_t(\omega)$ for each $\omega \in \Oany$.
\end{itemize}
\end{mydefinition}

Below we give a soundness proof of SSMs as stochastic processes. We found that we need a different proof strategy from that for antitone ranking functions~\cite{Kenyon-RobertsO21}, where the use of an antitone function was first proposed. To the best of our understanding, the proof in~\cite{Kenyon-RobertsO21} relies on the fact that the ranking condition is imposed at every non-terminal state, while in SSM we also need to reason about the behavior at $\Oany$, where we only impose the unaffecting condition.
\begin{mytheorem}\label{thm:soundnessSSM}
    Suppose a probability space $(\Omega, \calF, \bbP)$ admits a Streett supermartingale for a Streett pair $\bigl( \seqOf{\Ofin},\seqOf{\Oinf} \bigr)$. 
    Then we have
    \begin{align}
        \bbP\bigl(( \overset{\infty}{\forall}t.\omega\not\in\Ofin) \lor ( \overset{\infty}{\exists}t.\omega\in\Oinf) \bigr)=1.%
    \end{align}
\end{mytheorem}

\begin{proof}
Let $\seqOf{X_t}$ be the witnessing Streett SM. Assume the contrary and suppose $\bbP\bigl(( \overset{\infty}{\exists}t.\omega\in\Ofin) \land ( \overset{\infty}{\forall}t.\omega\not\in\Oinf) \bigr)>0$, in which case, there exists $k \in \bbN$ and $N\in \bbR$ such that 
$\bbP\bigl(( \overset{\infty}{\exists}t.\omega\in\Ofin) \land ( \forall t \geq k.\omega\not\in\Oinf) \land X_k \leq N \bigr)>0$. 

Let $\Phi_t(\omega) \equiv \forall t' \leq t. \omega \not\in \Oinf[t'+k]$. 
Then define a stopping time $T$ over the filtration $\seqOf{\calG_t} = \seqOf{\calF_{t+k}}$ by $T(\omega) = 0$ if $X_k(\omega) >N$; or otherwise $T(\omega) = \inf\{t\in \bbN \mid \lnot \Phi_t(\omega)\}$. 
Define a stochastic process $\seqOf{Y_t}$ by 
$Y_t(\omega) = 0$ for $T(\omega)=0$; and for $T(\omega)>0$, let $Y_t(\omega) = X_{t+k}(\omega)$ for each $t\leq T(\omega)$, and $Y_t(\omega) = Y_{T(\omega)}(\omega)$ for each $t> T(\omega)$. 

We have $Y_t(\omega) \geq 0$ for each $\omega$ and $t$; we also have $\bbE[Y_0]\leq N$. We have
\begin{align}
    \int_{\sem{T>t}}Y_{t+1}d\bbP = 
    \int_{\sem{T>t}}\bbE[Y_{t+1} \mid \calG_{t}]d\bbP \leq  \int_{\sem{T>t}} Y_t - \mathbf{1}_{\Ofin[t+k]}\cdot c(Y_t)d\bbP.
\end{align}
We also have $Y_{t+1}(\omega) = Y_t(\omega)$ for $\omega \in \sem{T\leq t}$, so overall we have
\begin{align}
    \bbE[Y_{t+1}] \leq \bbE[Y_t]-\int_{\sem{T>t} \cap \Ofin[t+k]}c(Y_t)d\bbP.
\end{align}
By iteratively applying the inequality above, we have the following at the limit:
\begin{align}
    \sum_{t\in \bbN} \int_{\sem{T>t} \cap \Ofin[t+k]}c(Y_t)d\bbP \leq \bbE[Y_0]\leq N.\label{eq:contradictionSource}
\end{align}
Now observe that $\seqOf{Y_t}$ is a supermartingale (i.e., $\bbE[Y_{t+1} \mid \calG_t](\omega) \leq Y_t(\omega)$ for each $\omega$ and $t$ and $\bbE[Y_0] < \infty$) because it is stopped at the first occurrence of $\omega \in \Oinf[t+k]$, and until then, we have $\omega \in (\Ofin[t+k] \setminus \Oinf[t+k]) \cup \Oany[t+k]$. Also recall $Y_t(\omega) \geq 0$ for each $\omega$ and $t$. Thus by Doob's convergence theorem, there exists a random variable $Y_\infty$ such that $\lim_{t \to \infty} Y_t (\omega)= Y_\infty(\omega) $ ($\bbP$-a.s.) and $\bbE[Y_\infty] < \infty$ (thus in particular, $Y_\infty(\omega)<\infty$, $\bbP$-a.s.). 
This implies that $\inf_{t \in \bbN}c(Y_t(\omega)) > 0$ holds $\bbP$-a.s., %
and thus the LHS of~(\ref{eq:contradictionSource}) must be infinite, which is a contradiction. 
More formally, fix any $\varepsilon >0$, and by Doob's theorem we have
\[
\overset{\infty}{\forall}t. Y_t(\omega) \leq Y_\infty(\omega)+\varepsilon <\infty \qquad \mbox{($\bbP$-a.s.)}
\]
and thus, by the non-increasing property and strict positivity of $c$, we have
\[
\overset{\infty}{\forall}t. c(Y_t(\omega)) \geq c(Y_\infty(\omega)+\varepsilon)>0, \qquad \mbox{($\bbP$-a.s.)}
\]
meaning that $c(Y_t(\omega))$ has a strictly positive lower bound for all sufficiently large $t$. 
Hence we have%
\[
\overset{\infty}{\exists}t.\omega \in \Ofin \implies 
\sum_{t\in\bbN} \mathbf{1}_{\Ofin[t+k]}(\omega)\cdot c(Y_t(\omega)) = \infty \qquad \mbox{($\bbP$-a.s.)}
\]
Let $\Psi(\omega) \equiv T(\omega)=\infty \land \overset{\infty}{\exists}t.\omega \in \Ofin$. 
Recall that we have $\bbP(\Psi) >0$ by assumption, and thus we have $\int_{\sem{\Psi}}\sum_{t\in\bbN} \mathbf{1}_{\Ofin[t+k]}(\omega)\cdot c(Y_t(\omega))d\bbP(\omega) = \infty$. 
By Tonelli's theorem, the integral and sum over non-negative functions are interchangeable, and we have
\begin{align*}
    \int_{\sem{\Psi}}\sum_{t\in\bbN} \mathbf{1}_{\Ofin[t+k]}\cdot c(Y_t)d\bbP &= \sum_{t\in\bbN} \int_{\sem{\Psi}}\mathbf{1}_{\Ofin[t+k]}\cdot c(Y_t)d\bbP \\
    &\leq \sum_{t\in\bbN} \int_{\sem{T>t}}\mathbf{1}_{\Ofin[t+k]}\cdot c(Y_t)d\bbP \\
    &\leq N <\infty,
\end{align*}
which is a contradiction.
\end{proof}

Having Theorem~\ref{thm:soundnessSSM}, the soundness of SSM maps can be proved via a standard argument.
\begin{proof}[Proof of Theorem~\ref{thm:SoundnessStrongSSMMap}]
For any initial state $s_I$ and deterministic scheduler $\sigma$ of $\calM$, let $\seqOf{X_t}$ be the stochastic process induced from $\eta$ under $s_I$ and scheduler $\sigma$: recall the underlying probability space $(\Omega,\calF,\bbP)$ is the dynamics of $\calM$ under $s_I$ and $\sigma$. Then it is easy to check that 
$\seqOf{X_t}$ is an SSM for the Streett pair $\bigl( \seqOf{\Ofin},\seqOf{\Oinf} \bigr)$, where $\Ofin = \sem{s_t \in \Sfin}$ and $\Oinf = \sem{s_t \in \Sinf}$. 
Therefore, by Theorem~\ref{thm:soundnessSSM}, $\omega \in \Omega$ satisfies the Streett condition $\bigl( \seqOf{\Ofin},\seqOf{\Oinf} \bigr)$ $\bbP$-a.s., meaning that a run of $\calM$ under $s_I$ and $\sigma$ satisfies the Streett condition $(\Sfin, \Sinf)$ almost surely.%
\end{proof}

\subsection{A Key Inequality for the Soundness Proof}
Here we show another key technical component for our proof, which sets up a use of relative well-behavedness in the proof.
\begin{mylemma}\label{lem:keyIneq}
    Let $p\in \calD(\bbR)$ be relatively well-behaved over the set of all polynomial functions. 
    Then for any $k\in \bbN$, 
    there exist $\theta \in (0,1)$ and $P \in (0,1)$ such that, for any $f \in \Poly_k$ that satisfies $p(f<0) \leq P$,
    \begin{align}
        \int fdp \geq 0 
        \qquad \mbox{and} \qquad
        \int(f^+)^\theta dp \leq 
        \biggl(
        \int fdp 
        \biggr)^\theta. \label{eq:keyIneqGeometricContraction}
    \end{align}
\end{mylemma}

At a high level, the use of the lemma is as follows. 
In our soundness proof, for a given lazy SSM $\eta$, we show that the function $\tilde
\eta = (\eta^+)^\theta$ is a strongly non-negative SSM for some $\theta \in (0,1)$.
The inequality (\ref{eq:keyIneqGeometricContraction}) lets us translate the ranking condition of $\eta$ to that of $\tilde\eta$ by a chain of inequalities
\begin{align*}
    \bbX \tilde\eta (s) \leq (\bbX\eta(s))^\theta \leq (\eta(s)-c(\eta(s)))^\theta = \tilde\eta(s)- \hat{c}(\tilde\eta(s)),
\end{align*}
where the function $\hat{c}(x) = x-(x^{\frac{1}{\theta}}-c(x^{\frac{1}{\theta}}))^\theta$ is again an antitone function.

\begin{proof}
    Take any $f \in \Poly_k$, and write $\pneg = p(f <0)$ for brevity. 
    By Jensen's inequality, and by the fact $\int f^+dp = \int_{\sem{f\geq 0}} f^+dp$, we have the following for any %
    $\theta \in (0,1)$:
    \begin{align}
        \int(f^+)^\theta dp%
         \leq (1-\pneg)^{1-\theta}\biggl(
        \int f^+ dp
        \biggr)^\theta.
    \end{align}
    By the definition of relative well-behavedness, we have the following for $\pneg \leq C_1$:
    \begin{align}
        \int(f^+)^\theta dp
         \leq (1-\pneg)^{1-\theta} \cdot (1+C_2\cdot \pneg)^\theta\cdot
         \biggl(
        \int f dp
        \biggr)^\theta.
    \end{align}%
Recall $C_1$ and $C_2$ can be chosen independently of $f$. 
Thus it suffices to show that there exist $\theta \in (0,1)$ and $P \in (0,C_1]$ that satisfy the following for any $q \leq P$:
\begin{align}
    D(q) := (1-q)^{1-\theta} \cdot (1+C_2\cdot q)^\theta \leq 1.
\end{align}
We have $D(0) =1$, and
by elementary calculus we have 
\(\left.\frac{d}{dp}D(p)\right|_{p=0}<0\) when $\theta \in (0,1)$ is sufficiently small (see Appendix A). %
Thus the desired $\theta$ and $P$ exist.
\end{proof}

\subsection{Proof of Theorem~\ref{thm:SoundnessLSSMmap}}
Now we assemble all the pieces into our proof.
\begin{proof}[Proof of Theorem~\ref{thm:SoundnessLSSMmap}]

    Let $\eta$ be a polynomial lazy SSM as in the theorem. 
Take any $s_I \in S_I$ and any deterministic scheduler $\sigma$ of $\calM$, and let $(\Omega,\calF,\bbP)$ be the dynamics of $\calM$ under $s_I$ and $\sigma$. 
Recall that each element $\omega \in \Omega$ is a run of $\calM$. 
For a run $\omega=s_0s_1\ldots$, we use $s_t$ to denote its $t$-th state; when $s_t$ appears as a random variable, it denotes the coordinate map $s_t:\Omega\to S$ that sends a run to its $t$-th state. 
The intended meaning will be clear from context.

We write $p_t(\omega)$ for the probability that $\eta$ becomes negative after a transition from the finite path $s_0\ldots s_t$, that is,
  $p_t(\omega) = \mynext_\sigma\mychar_{\sem{\eta<0}}(s_0\ldots s_t)$.
Let $\seqOf{X_t}$ be the stochastic process induced by $\eta$, i.e., $X_t(\omega)=\eta(s_t)$. 
Finally, let $k$ be the maximum degree of the polynomials in $\{\eta_\ell \mid \ell \in L\}$.
    \newcommand{\Orec}{\Omega_\mathsf{rec}}
    
    Now let $\Orec = \sem{( \overset{\infty}{\forall}t.s_t\not\in\Sfin) \lor ( \overset{\infty}{\exists}t.s_t\in\Sinf)}$, thus we need to show $\bbP(\Orec) = 1$. 
     We first observe the following: for any $q\in (0,1)$, we have
     \begin{align}
         (\overset{\infty}{\exists}t. p_t(\omega) > q)
      \implies
      \omega \in \Orec
      \quad \mbox{($\bbP$-a.s.)} \label{eq:insertCondition}
     \end{align}
     This is derived from two observations. First, we have the following (if we keep tossing a coin, then we eventually observe the tail almost surely):
     \[
     (\overset{\infty}{\exists}t. p_t(\omega) > q)
      \implies
      \overset{\infty}{\exists}t. X_t(\omega) < 0 
      \quad \mbox{($\bbP$-a.s.)} 
     \]
     Second, we have the following:
     \[
     \omega \not \in \Orec
      \implies
      \overset{\infty}{\forall}t. X_t(\omega) \geq 0. 
      \quad \mbox{($\bbP$-a.s.)} 
     \]
     This is because $\omega \not\in \Orec$ implies
     (i) $\overset{\infty}{\forall}t.s_t\not\in\Sinf$, thus $\overset{\infty}{\forall}t.s_t \in \Sfin\cup\Sany$, meaning that the value $X_t(\omega)$ never comes back to non-negative once it gets negative at a  sufficiently large $t$ (by the non-negativity and stability at negativity); and
     (ii) $\overset{\infty}{\exists}t.s_t\in\Sfin$, which implies  $\overset{\infty}{\exists}t. X_t(\omega) \geq 0$ by non-negativity.
     
     Thus, by~(\ref{eq:insertCondition}), we have $\bbP(\Orec) = \bbP(\Orec \cup \sem{\overset{\infty}{\exists}t. p_t > q})$ 
     for any $q$. Hence, to prove $\bbP(\Orec) = 1$, 
     it suffices to show the following holds for some $q \in (0,1)$: 
    \begin{align}%
        \bbP(( \overset{\infty}{\forall}t.s_t\not\in\Sfin) \lor ( \overset{\infty}{\exists}t.s_t\in\Sinf \lor p_t > q)) = 1.\label{eq:whatWeProve}
    \end{align}
    We prove (\ref{eq:whatWeProve}) by showing that, for some $\theta \in (0,1)$ and $q \in (0,1)$, a stochastic process $\seqOf{Y_t}$ defined by $Y_t(\omega) = (\eta^+(s_t))^\theta$ is an SSM with $\Ofin=\sem{s_t\in \Sfin \land p_t\leq q}$, $\Oinf=\sem{s_t\in \Sinf \lor p_t >q}$, and $\Oany = \sem{s_t \in \Sany \land p_t \leq q}$. %
    Below, we take $\theta\in (0,1)$ and $q\in (0,1)$ under which any construct $p$ of $P$ satisfies the following: 
    for any $f\in \Poly_k$ such that $p(f<0) \leq q$, we have the inequalities (\ref{eq:keyIneqGeometricContraction}):
    \begin{align*}
        \int fdp \geq 0 
        \qquad \mbox{and} \qquad
        \int(f^+)^\theta dp \leq 
        \biggl(
        \int fdp 
        \biggr)^\theta. 
    \end{align*}
    Such $\theta$ and $q$ exist by Lemma~\ref{lem:keyIneq}, the relative well-behavedness of the distributions in the variable construct of $P$, and the finiteness of the variable construct. We write $\tilde{\eta} = (\eta^+)^\theta$.

    It is clear from the definition that 
    $\seqOf{Y_t}$ is non-negative. %
    We show  $\seqOf{Y_t}$ satisfies the ranking condition over $\seqOf{\Ofin}$ as follows. 
    By Proposition~\ref{prop:condExpConventional}, 
    we have 
    \[
    \bbE[Y_{t+1} \mid \calF_t](\omega) = \bbX_\sigma \tilde\eta(s_0 \cdots s_t) = \int \tilde\eta \ dP(s_t,\sigma(s_0 \cdots s_t)).
    \]
    Here, $\mu := P(s_t,\sigma(s_0 \cdots s_t))$ is a distribution that samples up to one variable or location according to a construct $p$ of $P$, and takes fixed values for the rest. That is, 
    we have either 
    \begin{enumerate}
        \item $\mu= \delta_{s}$ 
        for some $s\in S$,
        \item $\mu=  p \times\delta_{\myvec{x}}$ for some $\myvec{x} \in \bbR^n$ and a construct $p\in \calD(L)$ of $P$, or
        \item $\mu=  p \times\delta_{(\ell, \myvec{x}_{-i})}$ for some $\ell \in L$,  $\myvec{x}_{-i} \in \bbR^{n-1}$, and a construct $p\in \calD(\bbR)$ of $P$.
    \end{enumerate}
    Below we show, in either case, we have $\int \eta \ d\mu \geq 0$ and  $\int \tilde\eta \ d\mu \leq (\int \eta \ d\mu)^\theta$ when $\mu$ is given by $\omega$ such that  $p_t(\omega) \leq q$. 
    In Case 1 or 2, $p_t(\omega) \leq q$ implies $p_t(\omega)=0$ (because we take $q$ small enough), meaning $\mu(\eta <0) = 0$ and thus $\int \eta \ d\mu \geq 0$. 
    Because $\mu$ has finite support in these cases, let $\{s^{(1)}, \ldots, s^{(k)}\}$ be the support of $\mu$ with their probability masses $q_1, \ldots, q_k \in [0,1]$. Then we have
    \[
    \int \tilde\eta  \ d\mu = \sum_{i=1}^k q_i\cdot(\eta(s^{(i)}))^\theta \leq \biggl(\sum_{i=1}^k q_i\cdot\eta(s^{(i)})\biggr)^\theta = \biggl(\int \eta  \ d\mu\biggr)^\theta,
    \]
    where the inequality comes from Jensen's inequality.

    In Case 3, let $f := \eta_{\ell, \myvec{x}_{-i}}$ (recall $\eta_{\ell, \myvec{x}_{-i}}(x_i) = \eta(\ell,x_i,\myvec{x}_{-i})$), where $\ell$ and $ \myvec{x}_{-i}$ are as designated in the case condition. Then $f \in \Poly_k$ and $\int\tilde\eta \ d\mu = \int (f^+)^\theta \ dp$ and $\int \eta \ d\mu = \int f \ dp$, where $p$ is as designated in the case condition. 
    By $p_t(\omega) \leq q$, we have $\mu(\eta <0) = p(f <0) \leq q$ and thus the inequalities (\ref{eq:keyIneqGeometricContraction}) are applicable. 
    Therefore we have $\int \eta \ d\mu \geq 0$ and
    \[
    \int\tilde\eta \ d\mu = \int (f^+)^\theta \ dp \leq 
    \biggl(\int f \ dp\biggr)^\theta = \biggl(\int \eta \ d\mu\biggr)^\theta,
    \]
    where the inequality comes from (\ref{eq:keyIneqGeometricContraction}). 
    
    So far we have shown that %
    $p_t(\omega) \leq q$ 
    implies $\bbE[Y_{t+1} \mid \calF_t] (\omega) \leq (\int \eta \ d\mu)^\theta$: observe $\omega \in \Ofin$ satisfies $p_t(\omega) \leq q$. 
    We also recall that  $\omega \in \Ofin$ implies that the ranking condition $\int \eta \ d\mu \leq X_t(\omega) - c(X_t(\omega))$ is available. Overall, we have the following for $\omega \in \Ofin$: \begin{align*}
        \bbE[Y_{t+1} \mid \calF_t](\omega) 
        \leq \bigl(X_t(\omega)-c(X_t(\omega))\bigr)^\theta  = Y_t(\omega) - 
        \biggl(Y_t(\omega) -\bigl(X_t(\omega)-c(X_t(\omega))\bigr)^\theta \biggr).
    \end{align*}
    WLOG we assume $c$ is continuous (one can replace $c$ with $c'(x) = \int_x^{x+1} c(y)\,dy$, which is continuous, antitone, and satisfies $0<c'(x)\leq c(x)$ for all $x\geq 0$).
    Then there exists a unique $x_0>0$ such that $c(x_0)=x_0$; hence
    $X_t(\omega)-c(X_t(\omega))\ge 0$ iff $X_t(\omega)\ge x_0$
    (iff $Y_t(\omega)\ge x_0^\theta$).
    Therefore, we have $\bbE[Y_{t+1} \mid \calF_t](\omega) \leq Y_t(\omega)-\hat{c}(Y_t(\omega))$ for $\omega \in \Ofin$, where
    we define $\hat{c}$ as follows:
    \begin{align*}
        \hat{c}(y) =
        \begin{cases}
            y - (y^{\frac{1}{\theta}}-c(y^{\frac{1}{\theta}}))^\theta & \mbox{if $y \geq x_0^\theta$,}\\
            x_0^\theta & \mbox{if $y < x_0^\theta$.}
        \end{cases}
    \end{align*}
    It is straightforward by elementary calculus to show that $\hat{c}$ is antitone. 
    Thus $\seqOf{Y_t}$ satisfies the ranking condition over $\Ofin$ with an antitone function $\hat{c}$.

    Showing the unaffecting condition of $\seqOf{Y_t}$ over $\Oany$ is easy: for $\omega \in \Oany$, we have
    \begin{align*}
        \bbE[Y_{t+1} \mid \calF_t] (\omega) \leq
        \biggl(\int \eta \ d\mu\biggr)^\theta
        \leq \bigl(
        X_t (\omega)
        \bigr)^\theta
        = Y_t (\omega),
    \end{align*}
where we have shown the first inequality holds for $\omega \in \Oany \subseteq \sem{p_t\leq q}$; the second is due to the unaffecting condition of $\seqOf{X_t}$, that is, $\int \eta \ d\mu \leq X_t(\omega)$.
\end{proof}
\begin{myremark}\label{rem:whyAntitone}
    The antitone relaxation of the ranking condition is crucial in our proof. Notice that the antitone ranking condition is equivalent to the conventional, fixed-valued ranking condition when we let $c$ be a constant function. Even in that case, $\hat{c}$ in our proof is not a constant function---observe we have $\hat c(y) \to 0$ as $y \to \infty$. Intuitively, this is because the ranking condition of $\eta$ is ``diluted'' by taking power of $\theta$, and such a dilution is more severe when the value of $\eta(s)$ is larger.
    As a result, it remains open if a polynomial lazy SSM witnesses \emph{positive} recurrence when $c$ is a constant function: See also Remark~\ref{rem:positiveRec}. \demo
\end{myremark}
\begin{myremark}\label{rem:WhyArithmeticUnnecessary}
    The existing proof of lazy lexicographic RSMs~\cite{TakisakaZWL24} requires $\calM$ to be induced from a probabilistic program with linear arithmetic, but we found this requirement unnecessary. The existing proof is concerned about the existence of $\bbE[X_{t}]$, which is not guaranteed when the range of $\eta$ is unbounded from below: When $\bbE[X_{t}]$ does not exist, then $\bbE[X_{t} \mid \calF_{t-1}]$ also does not exist, which causes trouble in describing the ranking condition of $X_t$. %
    However, the argument can instead be made directly over the non-negative process $\seqOf{Y_t}$, eliminating the assumption.
    \demo  
\end{myremark}

\section{From Single-Dimensional SSMs to Multi-Dimensional Supermartingales}\label{sect:fromOneToLex}
In this section, we introduce several multi-dimensional supermartingales constructed from (strongly non-negative) SSMs and lazy SSMs. 
These multi-dimensional supermartingales %
are sequences of (lazy) SSMs that jointly certify the target $\omega$-regular property. 
We consider two joint-reasoning frameworks: \emph{lexicographic SSMs} and \emph{(lexicographic) progress-measure supermartingales} (PMSMs). 
We apply these frameworks to SSMs and lazy SSMs, 
yielding four supermartingale variants: 
\begin{itemize}
    \item \emph{Lexicographic SSMs} and \emph{(lexicographic) PMSMs}, which are 
non-deterministic and antitone adaptations of existing notions; 
    \item \emph{lazy lexicographic SSMs} (Definition~\ref{def:LazyLexSSMMap}) and \emph{lazy (lexicographic) PMSMs}, which are new notions that adopt weak non-negativity.
\end{itemize}

Lexicographic SSMs and their variants have been understood as single entities that follow martingale conditions over lexicographic ordering, rather than collections of one-dimensional SSMs. 
To our knowledge, the latter view is our new observation. 
We formally prove these two views are equivalent, taking non-negative SSM as an example.

\subsection{Lexicographic SSMs and Their Soundness}
We restate the decomposition rule of Streett conditions in Section~\ref{subsec:fromOneDimToLex} as an explicit proposition. 
Below, a \emph{disjoint covering} of a set $\calX$ is a sequence $\calX_1, \ldots, \calX_k$ of disjoint sets such that $\calX \subseteq \cup_i\calX_i$.
\begin{myproposition}\label{prop:StreettDecomposition}
    For a Streett pair $(\Sfin,\Sinf)$, let $S_1,\ldots, S_k \subseteq S \setminus \Sinf$ be a disjoint covering of $\Sfin$. 
    We define a Streett pair $(\Sfin^{(i)}, \Sinf^{(i)})$ as follows, for each $i \in \{1,\ldots,k\}$:
\begin{align}
    \Sfin^{(i)} = S_i \quad , \quad \Sinf^{(i)} = \Sinf \cup (S_1 \cup \cdots \cup S_{i-1}).\label{eq:StreettDecomposition}
\end{align}
Then the conjunction of $(\Sfin^{(1)}, \Sinf^{(1)}),\ldots,(\Sfin^{(k)}, \Sinf^{(k)})$ implies $(\Sfin,\Sinf)$; 
the converse also holds when $\Sfin = \cup_i S_i$.
\end{myproposition}

\begin{proof}
The claim is canonically true for $k=1$, and the argument in Section~\ref{subsec:fromOneDimToLex} proves 
the case $k=2$ with $S_2 = \Sfin \setminus S_1$. 
As we have $\Sfin \setminus S_1 \subseteq S_2$ for any $S_1$ and $S_2$ that satisfy the assumption,  $(\Sfin^{(2)}, \Sinf^{(2)})$ implies $(\Sfin \setminus S_1, \Sinf^{(2)})$, thus the case $k=2$ is proved. %
The general case is shown by induction: By induction hypothesis, 
 the conjunction of 
$(\Sfin^{(1)}, \Sinf^{(1)}),\ldots,(\Sfin^{(k-2)}, \Sinf^{(k-2)}),(\Sfin^{(k-1)}\cup\Sfin^{(k)}, \Sinf^{(k-1)})$ implies $(\Sfin,\Sinf)$;
then by the claim for $k=2$, the conjunction of 
$(\Sfin^{(k-1)}, \Sinf^{(k-1)})$ and $(\Sfin^{(k)}, \Sinf^{(k)})$ implies $(\Sfin^{(k-1)}\cup\Sfin^{(k)}, \Sinf^{(k-1)})$. We also have the equivalence claim through a similar analysis.
\end{proof}

\begin{mydefinition}[lazy lexicographic SSM map]\label{def:LazyLexSSMMap}
    Let $\calM=(S,A,P,S_I)$ be an MDP, and let $(\Sfin, \Sinf)$ be a Streett pair over $S$. %
    A sequence $\eta_1, \ldots, \eta_k$ of MMs is a \emph{lazy lexicographic Streett supermartingale map} for $(\Sfin, \Sinf)$ if 
    there is a disjoint covering $S_1,\ldots, S_k \subseteq S \setminus \Sinf$ of $\Sfin$ %
    such that,  for each $i \in \{1, \ldots, k\}$, the MM
    $\eta_i$ is a lazy SSM map (Definition~\ref{def:LazySSMMap}) for the Streett pair $(\Sfin^{(i)}, \Sinf^{(i)})$. Here, $\Sfin^{(i)}$ and $\Sinf^{(i)}$ are as in (\ref{eq:StreettDecomposition}), and each $\eta_i$ may refer to a different invariant $\inv_i$ and antitone function $c_i:\bbR_{\geq 0} \to \bbR_{>0}$.
\end{mydefinition}

Soundness of lazy lexicographic SSM is derived as a simple corollary of the soundness of lazy SSM (Theorem~\ref{thm:SoundnessLSSMmap}), as follows. 

\begin{mytheorem}\label{thm:SoundnessLazyLexSSMMap}

     Let $\calM=(S,A,P,S_I)$ be a program MDP such that every distribution in the variable construct of $P$ is relatively well-behaved over polynomial functions, and let $(\Sfin, \Sinf)$ be a Streett pair over $S$. 
     If $\calM$ admits a polynomial lazy lexicographic SSM map, then $\calM$ satisfies the Streett condition $(\Sfin, \Sinf)$ almost surely. 
\end{mytheorem}

\begin{proof}
    By Theorem~\ref{thm:SoundnessLSSMmap}, a run of $\calM$ satisfies the Streett condition $(\Sfin^{(i)}, \Sinf^{(i)})$ almost surely for each $i \in \{1, \ldots, k\}$ under any initial state $s_I$ and a scheduler $\sigma$. Thus by Proposition~\ref{prop:StreettDecomposition}, a run of $\calM$ satisfies $(\Sfin, \Sinf)$ almost surely under any initial state $s_I$ and a scheduler $\sigma$.
\end{proof}

In a similar way, we can define a lexicographic extension of strongly non-negative SSM.  
Soundness is immediately proved from Theorem~\ref{thm:SoundnessStrongSSMMap} and Proposition~\ref{prop:StreettDecomposition}. 
Later we prove the notion is a non-deterministic and antitone adaptation of an existing notion.
\begin{mydefinition}[lexicographic SSM map]\label{def:LexSSMMap}
    Let $\calM=(S,A,P,S_I)$ be an MDP, and let $(\Sfin, \Sinf)$ be a Streett pair over $S$. %
    A sequence $\eta_1, \ldots, \eta_k$ of MMs is a \emph{(strongly non-negative) lexicographic Streett supermartingale map} for $(\Sfin, \Sinf)$ if 
    there is a disjoint covering
$S_1,\ldots,S_k \subseteq S\setminus S_{\mathrm{Inf}}$
of $S_{\mathrm{Fin}}$ %
    such that
    $\eta_i$ is an SSM map (Definition~\ref{def:StrongSSMMap}) for the Streett pair $(\Sfin^{(i)}, \Sinf^{(i)})$ for each $i \in \{1, \ldots, k\}$, where $\Sfin^{(i)}$ and $\Sinf^{(i)}$ are as in (\ref{eq:StreettDecomposition}). Here, each $\eta_i$ may refer to a different invariant $\inv_i$ and %
    $c_i:\bbR_{\geq 0} \to \bbR_{>0}$.
\end{mydefinition}

\begin{mytheorem}\label{thm:SoundnessLexSSMMap}
     Let $\calM=(S,A,P,S_I)$ be an MDP, and let $(\Sfin, \Sinf)$ be a Streett pair over $S$. 
     If $\calM$ admits a strongly non-negative lexicographic SSM map, then $\calM$ satisfies the Streett condition $(\Sfin, \Sinf)$ almost surely. \qed
\end{mytheorem}

\subsection{Characterization Based on Lexicographic Ordering}\label{subsec:EquivalenceLex}
We recall the lexicographic ordering $\succeq$ for supermartingales (cf.~\cite{KuraU26}). 
For $\myvec{x}, \myvec{y} \in \bbR^k$, we let
\begin{align*}
    \myvec{x} \succ_{[i]} \myvec{y} &\iff \forall j <i. \ [x_j \geq y_j] \ \land \ x_i \geq y_i+1, \\
    \myvec{x} \succ \myvec{y} &\iff \exists i \in \{1,\ldots, k\}. \ [\myvec{x} \succ_{[i]} \myvec{y} ], \\
    \myvec{x} \geq \myvec{y} &\iff \forall i \in \{1,\ldots, k\}. \ [x_i \geq y_i ], \\
    \myvec{x} \succeq \myvec{y} &\iff \myvec{x} \succ \myvec{y} \ \lor \ \myvec{x} \geq \myvec{y}. 
\end{align*}

Below we define lexicographic SSM in a ``standard'' way, i.e., as a  multi-dimensional MM that follows the SSM condition w.r.t. the ordering $\succeq$. We also note the definition boils down to the definition for Markov chains~\cite{KuraU26} when we substitute ``$\overline{\bbX}$'' with ``$\bbX$''.
\begin{mydefinition}[order-based lexicographic SSM map]\label{def:ConventionalLexSSMMap}
    Let $\calM=(S,A,P,S_I)$ be an MDP, and let $(\Sfin, \Sinf)$ be a Streett pair over $S$. 
    A sequence $\eta_1, \ldots, \eta_k$ of MMs is an \emph{order-based lexicographic SSM map} for $(\Sfin, \Sinf)$ if it satisfies the following for some invariant $\inv$:
    \begin{itemize}
        \item 
        $\boldsymbol{\eta}(s) \succ \overline{\bbX}\boldsymbol{\eta}(s)$ for each $s \in \inv \cap \Sfin$,
        \item $\boldsymbol{\eta}(s) \succeq \overline{\bbX}\boldsymbol{\eta}(s)$ for each $s \in \inv\cap\Sany$,
        \item  $\boldsymbol{\eta}(s)\geq \boldsymbol{0}$ (i.e., $\eta_i(s) \geq 0$ for each $i\in\{1,\ldots, k\}$) for each $s \in \inv$.
        \item The set $\sem{\boldsymbol{\eta} \succ_{[i]} \overline{\bbX}\boldsymbol{\eta}}$ 
        for each $i\in\{1,\ldots, k\}$, and $\sem{\boldsymbol{\eta} \geq  \overline{\bbX}\boldsymbol{\eta}}$, are measurable.
    \end{itemize}
\end{mydefinition}

Now we have the following equivalence theorem. A high-level idea is that $S_i$ in the joint-reasoning definition corresponds to $\sem{\boldsymbol{\eta} \succ_{[i]} \overline{\bbX}\boldsymbol{\eta}} \setminus \Sinf$ in the order-based one.
\begin{mytheorem}\label{thm:EquivalenceLexSSMMapOrdered}
    Let  an MDP $\calM = (S,A,P,S_I)$ and a Streett pair $(\Sfin, \Sinf)$ over $S$ be given. 
    Then a sequence $\boldsymbol{\eta} = (\eta_1, \ldots, \eta_k)$ of MMs is an order-based lexicographic SSM map for $(\Sfin, \Sinf)$ iff it is a lexicographic SSM map for $(\Sfin, \Sinf)$ with the constant antitone function $c_i(x)= 1$.
\end{mytheorem}

\begin{proof}%
    Suppose $\boldsymbol{\eta}$ is a lexicographic SSM map with $c_i(x)=1$ and $\eta_i$ uses an invariant $\inv_i$ for each $i$, and let $\inv=\cap_i\inv_i$. 
    Then it can be shown that 
    $\boldsymbol{\eta}(s) \succ_{[i]} \overline{\bbX}\boldsymbol{\eta}(s)$ for each $s \in \inv \cap S_i$ 
    and  $\boldsymbol{\eta}(s) \geq \overline{\bbX}\boldsymbol{\eta}(s)$ for each $s \in \inv \cap (\Sany \setminus (S_1\cup \cdots \cup S_k))$.
    As $\Sfin \subseteq S_1\cup \cdots \cup S_k$, this implies 
    $\boldsymbol{\eta}(s) \succ\overline{\bbX}\boldsymbol{\eta}(s)$ for each $s \in \inv \cap \Sfin$ and 
    $\boldsymbol{\eta}(s) \succeq\overline{\bbX}\boldsymbol{\eta}(s)$ for each $s \in \inv \cap \Sany$. 
    Non-negativity is also easy to check. %
    
    If $\boldsymbol{\eta}$ is an order-based lexicographic SSM map, let $S_i' = \sem{\boldsymbol{\eta} \succ_{[i]} \overline{\bbX}\boldsymbol{\eta}} \setminus \Sinf$, and then let $S_i = S_i' \setminus (S_{i+1}'\cup \cdots \cup S_k')$.
    Then 
    we have $\overline\bbX\eta_i(s)\leq \eta_i(s) -1$ for each $s \in \inv \cap S_i$ and $\overline\bbX\eta_i(s)\leq \eta_i(s)$ for each $s\in \inv \cap \Sany^{(i)}$. 
    This proves that $\boldsymbol{\eta}$ is a lexicographic SSM map.
\end{proof}
\begin{myremark}
    We restricted our attention to the case of $c_i(x) = 1$ to present the idea in a simple form. The equivalence under general antitone $c_i$ can be easily shown by generalizing Theorem~\ref{thm:EquivalenceLexSSMMapOrdered} w.r.t. the following ordering instead of $\succ_{[i]}$:
    \[ \myvec{x} \succ_{[c_i]} \myvec{y} \iff \forall j <i. \ [x_j \geq y_j] \ \land \ x_i \geq y_i+c_i(x_i),\]
    and also modifying the following argument accordingly. \demo 
\end{myremark}

\subsection{Lexicographic Progress-Measure Supermartingales and Their Soundness}
A \emph{parity condition} over a state space of an MDP $\calM = (S,A,P,S_I)$ is defined w.r.t. a \emph{priority function} 
$\pt:S \rightarrow \{0,\ldots, d\}$, where, for technical convenience, we assume that $d \in \bbN$ is odd. 
The function $\pt$ assigns a \emph{priority} $\pt(s)$ to each $s \in S$. 
An execution path $\omega = s_0s_1\ldots$ of $\calM$ satisfies the parity condition w.r.t. $\pt$ if the smallest priority that occurs infinitely often in $\omega$ is even:
\[
\mathsf{Parity}_{\pt}(\omega) \iff \min \{i \mid \overset{\infty}{\exists}t. \pt(s_t) =i\} \mbox{ is even}.
\]
Let $S_i = \sem{\pt = i}$ 
and $S_{\leq i} = S_0 \cup \cdots \cup S_i$ for each $i \in \{0,\ldots,d\}$. 
It is known that $\mathsf{Parity}_{\pt}$ is equivalent to the conjunction of the following Streett conditions~\cite[Lemma 2.6]{KuraU26}:%
\[
(S_{1},S_{\leq 0}) \quad,\quad (S_{3},S_{\leq 2}) 
\quad,\quad
\ldots
\quad,\quad
(S_{2i+1},S_{\leq 2i})
\quad,\quad
\ldots
\quad,\quad
(S_{ d},S_{\leq d-1}). 
\]

Each of these Streett conditions can be verified by a Streett SM, or more generally, by a lexicographic SSM. 
The latter idea is formalized as a \emph{lexicographic progress-measure supermartingale} (lexicographic PMSM)~\cite{KuraU26}, and it has been shown that the notion is indeed understood as a probabilistic / lexicographic extension of \emph{parity progress measure}~\cite[Proposition 5.14]{KuraU26}. 
We give a non-deterministic / antitone adaptation of lexicographic PMSMs based on Definition~\ref{def:StrongSSMMap}, and its lazy variant.

\begin{mydefinition}[(lazy) lexicographic PMSM map]
     Let $\calM=(S,A,P,S_I)$ be an MDP, and let $\pt:S \rightarrow \{0,\ldots, d\}$ be a measurable priority function, where $d$ is odd. 
     Also let $S_i$ and $S_{\leq i}$ be as above.

     A sequence $\eta_1, \ldots, \eta_k$ of MMs is called a \emph{lexicographic progress-measure supermartingale map (lexicographic PMSM map)} for $\mathsf{Parity}_{\pt}$ if it is partitioned into subsequences
     $H_1, H_3, \ldots, H_d$
     such that, for each odd $i \in \{1,\ldots,d\}$, the subsequence
     $H_i = (\eta_{i,1}, \ldots, \eta_{i,k_i})$
     is a lexicographic SSM map (Definition~\ref{def:LexSSMMap}) for the Streett pair $(S_{i},S_{\leq i-1})$. 
     If each $H_i$ is instead a lazy lexicographic SSM map (Definition~\ref{def:LazyLexSSMMap}) for $(S_{i},S_{\leq i-1})$, then we call the sequence a \emph{lazy lexicographic PMSM map}.
     Here, the components may refer to different invariants $\inv_{i,j}$ and antitone functions $c_{i,j}:\bbR_{\geq 0} \to \bbR_{>0}$. 
\end{mydefinition}

We have the following soundness theorems, which are immediate from Theorems~\ref{thm:SoundnessLexSSMMap} and~\ref{thm:SoundnessLazyLexSSMMap}.
\begin{mytheorem}\label{thm:SoundnessLexPMSMMap} 
    Let $\calM=(S,A,P,S_I)$ be an MDP, and let $\pt:S \rightarrow \{0,\ldots, d\}$ be a measurable priority function, where $d$ is odd. 
    If $\calM$ admits a lexicographic PMSM map for $\mathsf{Parity}_{\pt}$, then $\calM$ satisfies the parity condition $\mathsf{Parity}_\pt$ almost surely. \qed
\end{mytheorem}

\begin{mytheorem}\label{thm:SoundnessLazyLexPMSMMap}
     Let $\calM=(S,A,P,S_I)$ be a program MDP such that every distribution in the variable construct of $P$ is relatively well-behaved over polynomial functions, and let $\pt:S \rightarrow \{0,\ldots, d\}$ be a measurable priority function, where $d$ is odd. 
     If $\calM$ admits a polynomial lazy lexicographic PMSM map for $\mathsf{Parity}_{\pt}$, then $\calM$ satisfies the parity condition $\mathsf{Parity}_\pt$ almost surely. \qed
\end{mytheorem}

\section{Automated Synthesis Algorithm}\label{sec:algorithm}
In this section we present our synthesis algorithm for polynomial lazy lexicographic SSM maps.  
It is a canonical adaptation of the algorithm for lazy lexicographic RSMs to the Streett supermartingale setting. We give an overview here: Pseudocode is given in Appendix B.

As an input, the algorithm receives 
\begin{inparaenum}[(i)]
    \item a program MDP $\calM$ induced from a polynomial probabilistic program with relatively well-behaved sampling distributions (Definition~\ref{def:weakWellBehaved}),
    \item an invariant of $\calM$, and
    \item a Streett condition $(\Sfin, \Sinf)$ to be verified.
\end{inparaenum}
We write $S_{\mathsf{Fin},\ell} = \Sfin \cap (\{\ell\} \times\bbR^n)$, thus $\Sfin = \cup_{\ell \in L} S_{\mathsf{Fin},\ell}$.%

At a high level, the algorithm
incrementally constructs a sequence of lazy SSMs in a best-effort manner: Recall how we constructed a sequence of Streett pairs $(\Sfin^{(1)}, \Sinf^{(1)}),\ldots,(\Sfin^{(k)}, \Sinf^{(k)})$ in Section~\ref{subsec:fromOneDimToLex}.
In the $i$-th iteration, the algorithm receives $(\Srem^{(i)}, \Sinf^{(i)})$ %
as the remaining target condition, where $\Srem^{(i)}$ is the ``remainder'' of $\Sfin$.
Then 
it attempts to synthesize a lazy SSM $\eta_i$ that witnesses 
$(\Sfin^{(i)}, \Sinf^{(i)})$, where $\Sfin^{(i)}$ covers $\Srem^{(i)}$ as much as possible. 
Then it pushes $(\Srem^{(i+1)}, \Sinf^{(i+1)}) = (\Srem^{(i)}\setminus \Sfin^{(i)}, \Sinf^{(i)}\cup \Sfin^{(i)})$ to the next iteration. 
The algorithm continues the loop until either it observes $\Srem^{(i+1)} = \emptyset$, in which case it returns $(\eta_1,\ldots,\eta_i)$ as a witnessing lazy lexicographic SSM; or it fails to find $\eta_i$ with a nonempty $\Sfin^{(i)}$, in which case, it reports a failure.

The synthesis of $\eta_i$ is attempted in two steps: at this point, let $L_i$ be the set of locations that are unranked yet (i.e., $\Srem^{(i)} = \cup_{\ell \in L_i} S_{\mathsf{Fin},\ell}$). 
First, the algorithm tries to find a strongly non-negative SSM that ranks as many regions $S_{\mathsf{Fin},\ell} \subseteq \Srem^{(i)}$ as possible. 
If it finds an SSM that ranks at least one $S_{\mathsf{Fin},\ell} $, then it appends the found SSM as $\eta_i$; 
otherwise, it fixes some $L' \subseteq L_i$ and tries to rank \emph{exactly} those which are in $L'$, by synthesizing either of the following. 
There, we let $\Sfin^{(i)} =\cup_{\ell \in L'} S_{\mathsf{Fin},\ell}$.
\begin{enumerate}[(a)]
    \item A lazy SSM over $(\Sfin^{(i)}, \Sinf^{(i)})$. We call it the ``STAB'' option, abbreviating \emph{stability at negativity}.\label{item:STAB}
    \item A lazy SSM over $(\Sfin^{(i)}, \Sinf^{(i)})$, but we require a stronger condition over $\Sany^{(i)}$, called the \emph{worst-case unaffecting condition}: at $s\in S$, it requires $\forall a\in A. \ [P(s,a)(\eta \leq \eta(s)) =1]$, i.e., the value of $\eta$ does not increase through a transition almost surely. Observe this condition implies \emph{stability at negativity}. We call it the ``WORST'' option.\label{item:WORST}
\end{enumerate}
The option (\ref{item:WORST}) is proposed in the original algorithm~\cite{TakisakaZWL24} to make the algorithm feasible for an LP solver. We consider both options in our implementation, using an SMT solver.

The second step above is iterative: We try different choices of $L'$ selected by a user-defined procedure. 
In our implementation, we use a heuristic that exploits the path relationships between the program locations; see Appendix~B. %
We need such an iterative procedure because the martingale conditions for a lazy SSM at $\Sfin$ and $\Sany$ are incomparable, unlike those for a strongly non-negative SSM.
Due to this peculiarity, it is not known if there is a selection procedure for $L'$ that makes the synthesis algorithm for lazy lexicographic RSMs polynomial-time and complete~\cite{TakisakaZWL24}: The same limitation also applies to our algorithm.

\section{Experiments}
\label{sec:evaluation}
We ran experiments to evaluate how lazy lexicographic SSMs improve the algorithmic
applicability of template-based synthesis algorithms.
To this end, we consider three variants of lexicographic SSM synthesizers, all using the constant antitone function $c(x)\equiv 1$.
\begin{enumerate}[(a)]
    \item \textit{STR}: This algorithm synthesizes a strongly non-negative lexicographic SSM
    (Definition~\ref{def:LexSSMMap}). We use it as our baseline.\label{item:Alg_STR}
    \item \textit{WORST}: This is our lazy lexicographic SSM synthesizer with the
    ``WORST'' option (cf. Section~\ref{sec:algorithm}).\label{item:Alg_MCLC}
    \item \textit{STAB}: This is our lazy lexicographic SSM synthesizer with the
    ``STAB'' option.\label{item:Alg_STAB}
\end{enumerate}
In practice, we run only (\ref{item:Alg_MCLC}) and (\ref{item:Alg_STAB}),
because the results for (\ref{item:Alg_STR}) can be extracted from these runs.
Indeed, (\ref{item:Alg_MCLC}) and (\ref{item:Alg_STAB}) return a strongly
non-negative lexicographic SSM exactly when (\ref{item:Alg_STR}) succeeds.
Thus, we report a success of (\ref{item:Alg_STR}) whenever %
they find a strongly non-negative
lexicographic SSM.

\begin{wrapfigure}{r}{0.34\textwidth}
  \centering
  \scriptsize
  \begin{tabular}{@{}ll@{}}
    \toprule
    Name & Formula \\
    \midrule
    RA & $(F\,\mathit{at}(l_{\mathit{term}})) \wedge (G\, x \ge 0)$ \\
    OV & $F\,(\mathit{at}(l_{\mathit{term}}) \vee x < -64 \vee x > 63)$ \\
    RC & $G\,F\,(x \ge 0)$ \\
    PR & $G\,(x < -5 \Rightarrow F\,(x > 0))$ \\
    \bottomrule
  \end{tabular}

\end{wrapfigure}

We generate our benchmarks by modifying non-probabilistic benchmarks from
the literature, following the experimental setup of~\cite{AgrawalCP18, TakisakaZWL24}.
The original non-probabilistic benchmarks are taken from the verification tasks
in LTL-VerP~\citep{chatterjee2024ltlverp}.
The suite consists of polynomial programs and four LTL specifications with
polynomial atomic predicates, as shown in the table on the right.
We did not include \emph{refutation} tasks whose ground truth is false, since
they are beyond the scope of lexicographic SSMs.

To exercise probabilistic semantics, we evaluate two probabilistic variants.
The \emph{prob-loop} variant inserts a probabilistic execute/skip branch before
each nontrivial transition, modeling loop bodies of the form
\textbf{if prob($p$) then body else skip fi}.
The \emph{prob-assignment} variant perturbs deterministic assignments with
bounded uniform samples, such as
$y \sim \textbf{Unif}[a,b]\ ;\ x \leftarrow f(x)+y$.

We built our implementation by modifying LTL-VerP~\citep{chatterjee2024ltlverp}, with a straightforward adaptation of its invariant-synthesis procedure to probabilistic programs. We use Spot~\citep{spot2016} to translate LTL specifications into Streett automata. Symbolic polynomial obligations are constructed with SymPy~\cite{Meurer2017SymPy} and discharged by a Z3~\cite{deMoura2008Z3} backend after Farkas/Handelman/Putinar-style coefficient encodings~\cite{chatterjee2016termination}. Each benchmark is run under universal nondeterminism. We consider polynomial template degrees $2$ and $3$ for both supermartingale and invariant synthesis, yielding four degree combinations. We set a 20-second timeout for each Z3 call and a 600-second timeout for each degree combination, subject to an overall timeout of 1800 seconds per benchmark.

\subsection{Results}
The statistics of the experimental results are shown in Table~\ref{tab:shared-lazy-backend-summary}. 
Benchmark-wise results are given in Appendix~C.

WORST and STAB substantially outperform STR:
among 170 benchmarks, STR succeeds in 88 cases (51.7\% success), WORST succeeds
in 128 cases (75.2\% success), and STAB succeeds in 122 cases (71.7\% success).
The union of the instances solved by WORST and STAB consists of 134 cases
(78.8\% success).
These numbers show that the benefit of relaxing the non-negativity requirement
is substantial.

Both WORST and STAB have unique successes, although the STAB constraints are,
in theory, no stronger than the WORST constraints.
We found that this happens when \emph{stability at negativity} makes the
constraint problem highly complex, causing Z3 to time out.
Apart from these few exceptions, the computation times of WORST and STAB are
mostly similar.

\begin{table}[t]
\centering
\caption{Summary of the experiment results. W+S shows the number of instances solved by either WORST or STAB. Average time is computed over solved instances using the fastest successful result for each benchmark.
}
\label{tab:shared-lazy-backend-summary}
\scriptsize
\setlength{\tabcolsep}{2.8pt}
\begin{tabular}{lccccccccc}
\toprule
Variant & Property & Benchmark count & \multicolumn{4}{c}{Verified instances} & \multicolumn{2}{c}{Unique successes} & Avg. time\\
\cmidrule(lr){4-7}\cmidrule(lr){8-9}
 & & & STR & WORST & STAB & W+S & WORST & STAB & (s)\\
\midrule
prob-loop & RA & 27 & 13 & 19 & 16 & 20 & 4 & 1 & 199.1\\
prob-loop & OV & 29 & 10 & 21 & 21 & 26 & 5 & 5 & 338.7\\
prob-loop & RC & 18 & 13 & 15 & 15 & 15 & 0 & 0 & 78.5\\
prob-loop & PR & 11 & 2 & 5 & 5 & 5 & 0 & 0 & 202.6\\
\midrule
 & Subtotal & 85 & 38 & 60 & 57 & 66 & 9 & 6 & --\\
\midrule
prob-asgn & RA & 27 & 18 & 22 & 20 & 22 & 2 & 0 & 75.4\\
prob-asgn & OV & 29 & 15 & 26 & 25 & 26 & 1 & 0 & 273.1\\
prob-asgn & RC & 18 & 14 & 15 & 15 & 15 & 0 & 0 & 71.2\\
prob-asgn & PR & 11 & 3 & 5 & 5 & 5 & 0 & 0 & 191.8\\
\midrule
 & Subtotal & 85 & 50 & 68 & 65 & 68 & 3 & 0 & --\\
\midrule
 & Total & 170 & 88 & 128 & 122 & 134 & 12 & 6 & --\\
\bottomrule
\end{tabular}
\end{table}

\section{Related Work}\label{sec:relWorks}
The works most relevant to ours are \cite{TakisakaZWL24} and \cite{KuraU26}.  
Our lazy lexicographic SSM (Definition~\ref{def:LazyLexSSMMap}) combines the core ideas of these two lines of work: 
weak non-negativity from lazy lexicographic RSMs~\cite{TakisakaZWL24}, and 
$\omega$-regular verification via Streett supermartingales~\cite{KuraU26}.
Regarding the latter, we note that the relaxation of the SSM condition in~\cite{KuraU26} from the original formulation of~\cite{AbateGR24} is crucial for the results in Section~\ref{sect:fromOneToLex}.

Weak non-negativity has been well recognized in the context of \emph{lexicographic ranking functions}~\cite{BradleyMS05,Ben-AmramG15}. 
The corresponding idea for supermartingales is still relatively new~\cite{ChatterjeeGNZZ23,TakisakaZWL24}. 
Lexicographic RSMs with weak non-negativity were first proposed in~\cite{ChatterjeeGNZZ23}, where the technical intricacy specific to the probabilistic setting was observed (cf. Figure~\ref{fig:CounterEx}). 
Subsequently, two generalizations were proposed in~\cite{TakisakaZWL24}: 
one is the notion of \emph{$\varepsilon$-fixable lexicographic RSMs}, which is generally sound but difficult to incorporate into synthesis algorithms; 
the other is the notion of \emph{lazy lexicographic RSMs}, whose non-negativity condition is even weaker than that of $\varepsilon$-fixable lexicographic RSMs and is also more amenable to automated synthesis, but whose soundness has so far been established only under a particular setting.
The present paper significantly extends the scope of the soundness theorem for the latter, as summarized in Table~\ref{tab:lazy-lex-comparison}.

Several works study the verification of almost-sure satisfaction of $\omega$-regular properties~\cite{AbateGR24,HenzingerMSZ25,KuraU26}. 
\emph{Streett supermartingales} were introduced in~\cite{AbateGR24}; their martingale constraint was later relaxed in~\cite{KuraU26}, yielding the baseline for our $\omega$-regular verification framework. 
These works assume strong non-negativity. 
Other studies address quantitative verification for $\omega$-regular properties~\cite{AbateGR25,HenzingerMSZ25}, which likewise assume strong non-negativity. 
Relaxing this requirement is an interesting direction, but it is less clear than in the almost-sure setting whether and how such a relaxation is possible: 
since martingales for quantitative verification typically provide bounds on the quantities of interest, weakening non-negativity may change the meaning of the resulting certificate.

\section{Conclusion}\label{sec:conclusion}
We introduced \emph{lazy Streett supermartingales}, weakly non-negative martingale certificates for almost-sure $\omega$-regular verification.
We proved their soundness under less restrictive conditions than the prior lazy ranking supermartingale framework.
We also presented a new proof strategy that lifts the soundness of one-dimensional SSMs to their lexicographic variants.
Our synthesis experiments show that relaxing strong non-negativity substantially improves the success rate of template-based $\omega$-regular verification.

\bibliographystyle{ACM-Reference-Format}
\bibliography{ref}

\newpage
\appendix

\section*{Appendix}
\section{Omitted Details of Section~\ref{sec:LazySSMSoundness}}\label{append:LazySSMSoundness}

\subsection{Omitted Details of Proof of Lemma~\ref{lem:keyIneq}}
A concrete calculation is as follows.
First, 
we have
\[
\log D(q)
=
(1-\theta)\log(1-q)
+
\theta \log(1+C_2\cdot q).
\]
Differentiating with respect to \(q\), we obtain
\[
\frac{d}{dq}\log D(q)
=
-\frac{1-\theta}{1-q}
+
\frac{\theta\cdot C_2}{1+q\cdot C_2}.
\]
In particular, at \(q=0\),
\[
\left.\frac{d}{dq}\log D(q)\right|_{q=0}
=
-(1-\theta)+\theta \cdot C_2
=
\theta(C_2+1)-1.
\]
Thus we have 
\(\left.\frac{d}{dq}\log D(q)\right|_{q=0}<0\) when $\theta < \frac{1}{C_2+1}$, 
in which case, we also have 
\(\left.\frac{d}{dq}D(q)\right|_{q=0}<0\).
\begin{algorithm}[t]%
   \caption{Synthesis algorithm for polynomial lazy lexicographic SSM maps}\label{alg:Alg1}
   \textbf{Input:} A program MDP $\calM=(S,A,P,S_I)$ induced from a polynomial PP with relatively well-behaved sampling distributions, a Streett pair $(\Sfin^*, \Sinf^*)$, an invariant $I$\;
   Initialize $
   U \leftarrow \{\ell\in L \mid S_{\mathsf{Fin},\ell}\neq\emptyset\};%
   \ d \leftarrow 0; \ \Srem \leftarrow \Sfin^*;  \ \Sinf \leftarrow \Sinf^*$; 
   \\
   \While{$U$ is not empty}{
     $d \leftarrow d+1$;$ranked \leftarrow False$\;
     Construct and solve $\mathsf{Synth}_{U}^1$\;\label{line:consLP1}
     \If{there is no solution to $\mathsf{Synth}_{U}^1$}{\label{line:IfNoSol}
        \For{each $\calT \in \mbox{Class}(U)$}{\label{line:chooseTransitionsToRank}
            Construct and solve $\mathsf{Synth}_{U,\calT}^2$\;
            \If{a solution exists for $\mathsf{Synth}_{U,\calT}^2$}{
                $\eta_d \leftarrow $ RF from Solution of $\mathsf{Synth}_{U,\calT}^2$;$ranked \leftarrow True$\;\label{line:solvedsingle}
                $U \leftarrow U \backslash \calT; \; 
                \Sinf \leftarrow \Sinf \cup  ((\calT \times \bbR^n) \cap \Srem); \; 
                \Srem \leftarrow \Srem  \setminus (\calT \times \bbR^n);$ \label{line:delsingle}
                break\;\label{line:break}
            }
        }
        \If{not $ranked$}{\label{line:ifnosingle}
            {\bf Return} FALSE\;\label{line:fail}
            }
    }
     \Else{
         $\eta_d \leftarrow $ ranking function from Solution of $\mathsf{Synth}_{U}^1$\;\label{line:solved}
         $U \leftarrow U \backslash U'; \; 
         \Sinf \leftarrow \Sinf \cup  ((U' \times \bbR^n) \cap \Srem); \; 
         \Srem \leftarrow \Srem  \setminus (U' \times \bbR^n)$\;
         \label{line:del}
     }
   }
   {\bf Return} $(\eta_1,\eta_2,\cdots,\eta_d)$
\end{algorithm}

\begin{algorithm}[t]%
   \caption{Construction of lazy candidate classes}\label{alg:LazyClass}
   \textbf{Input:} an unresolved node set $U$, product graph $\mathcal{P}$, lazy option $b\in\{\mathsf{WORST},\mathsf{STAB}\}$\;\label{line:classInput}
   \textbf{Output:} an ordered list $\mathcal{L}$ of nonempty subsets of $U$\;\label{line:classOutput}
   $\mathcal{L} \leftarrow \{U\}$\;\label{line:classFull}
   $\mathcal{L} \leftarrow \mathcal{L} \cup \mathsf{Src}(U) \cup \mathsf{Succ}(U) \cup \mathsf{Shape}(U) \cup \mathsf{Split}(U)$\;\label{line:classGroups}
   \If{$b=\mathsf{WORST}$}{
        $\mathcal{L} \leftarrow \mathcal{L} \cup \{\{u\}\mid u\in U\}$\;\label{line:classSingleton}
   }
   remove empty and duplicate sets from $\mathcal{L}$\;\label{line:classDedupe}
   order $\mathcal{L}$ by estimated constraint size, breaking ties in favor of larger sets\;\label{line:classOrder}
   {\bf Return} $\mathcal{L}$\;\label{line:classReturn}
\end{algorithm}

\section{Omitted Details of Section~\ref{sec:algorithm}}\label{append:algorithm} 
\subsection{Line-by-Line Explanation of the Synthesis Algorithm}
The pseudocode of our algorithm is given in Algorithm~\ref{alg:Alg1}, %
whose summary is as follows. %
Similar to existing lexicographic RSM synthesis algorithms~\cite{ChatterjeeGNZZ23,AgrawalCP18}, 
it constructs a lexicographic SSM $(\eta_1,\cdots,\eta_d)$ 
in an iterative way. %
At the $d$-th iteration, the algorithm attempts to construct $\eta_d$ that ranks states in $(U \times \bbR^n) \cap \Srem$, i.e., those which are not ranked by $\eta_1, \ldots, \eta_{d-1}$ (Lines~\ref{line:consLP1}--\ref{line:break}). 
It first tries to construct such $\eta_d$ as a strongly non-negative SSM (Definition~\ref{def:StrongSSMMap}). 
This is done by solving the problem $\mathsf{Synth}_U^1$ (Line~\ref{line:consLP1}), 
which looks for a one-dimensional MM $\eta$ that satisfies the following with nonempty $U' \subseteq U$ that is as large as possible: Here, we let $S' = (U'\times \bbR^n) \cap \Srem$.%
\begin{enumerate}[(a)]%
\item for every $s \in S' \cap I$, %
we have $\mathsf{Rank}_\eta(s)$;%
\item for every 
$s \in (S \setminus (S' \cup \Sinf)) \cap I$, %
we have $\mathsf{UnAffect}_\eta(s)$;%
\item for every 
$s \in I$, we have $\mathsf{NonNeg}_\eta(s)$.
\end{enumerate}
The problem $\mathsf{Synth}_U^1$ is obtained by the reduction of conditions (a--c) via Farkas/Handelman/Putinar-style coefficient encodings~\cite{chatterjee2016termination}. 
If the solution $\eta$ of $\mathsf{Synth}_U^1$ is found for nonempty $U'$, then 
the algorithm lets $\eta_d=\eta$, adds it to the output, and updates the ``remainder'' Streett condition accordingly (Lines~\ref{line:solved}--\ref{line:del}); 
otherwise, %
it tries to construct $\eta_d$ under either ``WORST'' or ``STAB'' option (cf. Section~\ref{sec:algorithm}). %
This is done by solving $\mathsf{Synth}_{U,\calT}^2$ for each $\calT\in\mbox{Class}(U) $ (line~\ref{line:chooseTransitionsToRank}). 
The ordered list $\mbox{Class}(U)$ is constructed by Algorithm~\ref{alg:LazyClass}. 
The algorithm is described as follows: $\calT$ is a subset of the unresolved node set $U$. 
For a node $u$, let $\mathsf{loc}(u)$ be its program location, $\mathsf{Post}(u)$ be its one-step successor nodes in the product graph, and $\mathsf{sig}(u)$ be the local signature of the outgoing transitions from $u$ (target locations, labels, polynomial updates, nondeterministic group, and sampling distributions).
Lines~\ref{line:classInput}-\ref{line:classOutput} specify that Algorithm~\ref{alg:LazyClass} takes the current unresolved node set, the product graph, and the chosen lazy option, and returns a finite ordered list of nonempty candidates. 
The four node families used in line~\ref{line:classGroups} are:
\[
\begin{array}{rcl}
\mathsf{Src}(U) &=& \br{\br{u\in U\mid \mathsf{loc}(u)=\ell}\mid \ell \mbox{ is a program location}},\\[0.2em]
\mathsf{Succ}(U) &=& \br{\br{u'\in U\mid u'\in \br{u}\cup \mathsf{Post}(u)}\mid u\in U},\\[0.2em]
\mathsf{Shape}(U) &=& \br{\br{u\in U\mid \mathsf{sig}(u)=\sigma}\mid \sigma \mbox{ is a local transition signature}},\\[0.2em]
\mathsf{Split}(U) &=& \mbox{the non-singleton sets obtained by recursively splitting a fixed ordering of }U.
\end{array}
\]
Here $\mathsf{Src}(U)$ groups nodes at the same control-flow location, $\mathsf{Succ}(U)$ groups nodes that are coupled through one successor step, $\mathsf{Shape}(U)$ groups nodes with the same local probabilistic shape, and $\mathsf{Split}(U)$ provides deterministic smaller fallback candidates.
Line~\ref{line:classFull} of Algorithm~\ref{alg:LazyClass} first inserts the coarsest candidate $U$ itself. 
Line~\ref{line:classGroups} then appends the four structured families above. 
Line~\ref{line:classSingleton} adds singleton candidates for the ``WORST'' option, giving it a final node-by-node fallback when larger candidates are too difficult.
Line~\ref{line:classDedupe} removes empty and duplicate candidates. 
Line~\ref{line:classOrder} orders the remaining list by the estimated size of the generated constraint system, using larger candidates as a tie-breaker so that one successful component can remove more nodes from $U$.
Finally, line~\ref{line:classReturn} returns this ordered list to line~\ref{line:chooseTransitionsToRank} of Algorithm~\ref{alg:Alg1}.
For each candidate $\calT$ in this list, the problem $\mathsf{Synth}_{U,\calT}^2$ looks for $\eta$ that \emph{exactly} ranks states in $(\calT \times \bbR^n)\cap \Srem$, that is, in the ``STAB'' option,
\begin{enumerate}[(a')]%
\item for every 
$s \in S_\calT \cap I$, where $S_\calT =(\calT \times \bbR^n)\cap \Srem$, %
we have $\mathsf{Rank}_\eta(s)$ and $\mathsf{NonNeg}_\eta(s)$;
\item for every %
$s \in (S \setminus (S_\calT \cup \Sinf)) \cap I$,
we have $\mathsf{UnAffect}_\eta(s)$ and  $\mathsf{StabNeg}_\eta(s)$;
\end{enumerate}
or in ``WORST'' option, 
\begin{enumerate}[(a')]%
\item for every %
$s \in S_\calT \cap I$, 
we have $\mathsf{Rank}_\eta(s)$ and $\mathsf{NonNeg}_\eta(s)$;
\item for every %
$s \in (S \setminus (S_\calT \cup \Sinf)) \cap I$, 
we have $\forall a\in A. \ [P(s,a)(\eta \leq \eta(s)) =1]$ (the worst-case unaffecting condition).
\end{enumerate}

For each candidate $\calT$ returned by Algorithm~\ref{alg:LazyClass}, line~\ref{line:chooseTransitionsToRank} keeps the current unresolved set $U$ fixed and constructs the corresponding lazy query. 
If $\mathsf{Synth}_{U,\calT}^2$ has no solution, or if the solver cannot find one within the configured budget, the algorithm discards only this candidate $\calT$ and tries the next candidate in the ordered list; no node is removed from $U$ in this case.
Once $\mathsf{Synth}_{U,\calT}^2$ is solved for some $\calT$, line~\ref{line:solvedsingle} extracts the ranking function $\eta_d$ from the solution and marks the current dimension as successful. 
Line~\ref{line:delsingle} then removes exactly the nodes in $\calT$ from $U$ and moves the corresponding region from $\Srem$ to $\Sinf$. 
Line~\ref{line:break} stops the search over later candidates in the same iteration; the algorithm does not keep searching for a different candidate after the first satisfiable one is found. 
If it fails to solve $\mathsf{Synth}_{U,\calT}^2$ for every $\calT\in\mbox{Class}(U)$, then it concludes a failure and terminates (lines~\ref{line:ifnosingle}--\ref{line:fail}). 
If $\eta_d$ is computed (line~\ref{line:solvedsingle} or~\ref{line:solved}), the algorithm goes to the next iteration after updating $U$; the iteration continues until $U$ is empty. \todo{formal description about soundness/completeness of the algorithm is as follows.}%

\newpage
\section{Detailed Experimental Data}
\label{append:fullExpDetails}

Tables~\ref{tab:detail-loop-backends} and~\ref{tab:detail-assignment-backends} report the per-instance results for both lazy-negativity backends (WORST/STAB). Each backend entry is the best result for that backend over the configured template grid.  We write \Strong{d}{t} for a strongly non-negative certificate of dimension $d$ found in $t$ seconds and \Lazy{d}{t} for a lazy certificate; \Fail{} means that no certificate was found within the configured budget and \NA{} means that the benchmark/property pair was not included in the paper-aligned obligation set (i.e., the instance was a refutation task).

\begingroup
\scriptsize
\setlength{\tabcolsep}{1.4pt}
\renewcommand{\arraystretch}{1.06}

\begin{longtable}{>{\RaggedRight\arraybackslash}p{0.24\textwidth}*{8}{>{\centering\arraybackslash}p{0.085\textwidth}}}
\caption{Per-instance WORST and STAB results for the probabilistic-loop variant.}\label{tab:detail-loop-backends}\\
\toprule
Benchmark & \multicolumn{4}{c}{WORST} & \multicolumn{4}{c}{STAB}\\
\cmidrule(lr){2-5}\cmidrule(lr){6-9}
 & RA & OV & RC & PR & RA & OV & RC & PR\\
\midrule
\endfirsthead
\toprule
Benchmark & \multicolumn{4}{c}{WORST} & \multicolumn{4}{c}{STAB}\\
\cmidrule(lr){2-5}\cmidrule(lr){6-9}
 & RA & OV & RC & PR & RA & OV & RC & PR\\
\midrule
\endhead
\midrule
\multicolumn{9}{r}{\emph{Continued on next page}}\\
\endfoot
\bottomrule
\endlastfoot
\path{bresenham-ll.c.t2} & \Strong{4}{13.17} & \Lazy{4}{446.92} & \NA & \NA & \Strong{4}{13.17} & \Fail & \NA & \NA\\
\path{ChawdharyCookGulwaniSagivYang-ESOP2008-aaron12_true-termination.c.t2} & \NA & \Lazy{2}{432.10} & \NA & \NA & \NA & \Lazy{2}{429.87} & \NA & \NA\\
\path{cohencu-ll.c.t2} & \Strong{3}{5.02} & \Lazy{3}{207.01} & \NA & \NA & \Strong{3}{5.02} & \Fail & \NA & \NA\\
\path{cohendiv-ll.c.t2} & \Lazy{3}{573.38} & \Lazy{5}{638.66} & \Strong{1}{30.05} & \NA & \Fail & \Lazy{5}{639.78} & \Strong{1}{30.05} & \NA\\
\path{ComplInterv.c.t2} & \NA & \Fail & \NA & \NA & \NA & \Lazy{1}{273.76} & \NA & \NA\\
\path{dijkstra-u.c.t2} & \Lazy{12}{933.56} & \Fail & \Lazy{4}{580.26} & \NA & \Fail & \Lazy{1}{595.15} & \Lazy{4}{580.26} & \NA\\
\path{divbin__property.c.t2} & \Fail & \NA & \NA & \NA & \Fail & \NA & \NA & \NA\\
\path{divbin.c.t2} & \NA & \Lazy{8}{919.41} & \NA & \NA & \NA & \Lazy{8}{919.40} & \NA & \NA\\
\path{DoubleNeg.c.t2} & \Strong{2}{3.18} & \Strong{2}{4.39} & \NA & \NA & \Strong{2}{3.18} & \Strong{2}{4.39} & \NA & \NA\\
\path{egcd-ll.c.t2} & \NA & \Fail & \NA & \NA & \NA & \Fail & \NA & \NA\\
\path{egcd2-ll.c.t2} & \Lazy{5}{594.52} & \NA & \NA & \Fail & \Fail & \NA & \NA & \Fail\\
\path{egcd3-ll.c.t2} & \Fail & \NA & \NA & \Fail & \Fail & \NA & \NA & \Fail\\
\path{Factorial__property.c.t2} & \NA & \Fail & \NA & \NA & \NA & \Fail & \NA & \NA\\
\path{Factorial.c.t2} & \NA & \NA & \NA & \Lazy{2}{5.89} & \NA & \NA & \NA & \Lazy{2}{5.88}\\
\path{geo1-ll.c.t2} & \Strong{3}{3.32} & \Fail & \Strong{1}{5.95} & \Strong{2}{6.31} & \Strong{3}{3.32} & \Lazy{1}{550.56} & \Strong{1}{5.95} & \Strong{2}{6.31}\\
\path{geo2-ll.c.t2} & \Strong{3}{4.82} & \Strong{3}{6.98} & \Strong{1}{2.96} & \Lazy{1}{377.32} & \Strong{3}{4.82} & \Strong{3}{6.98} & \Strong{1}{2.96} & \Lazy{1}{377.28}\\
\path{geo3-ll.c.t2} & \Lazy{2}{358.61} & \Fail & \NA & \NA & \Lazy{2}{358.94} & \Lazy{1}{451.88} & \NA & \NA\\
\path{hard-ll__property.c.t2} & \NA & \Fail & \NA & \NA & \NA & \Lazy{1}{800.61} & \NA & \NA\\
\path{hard-ll.c.t2} & \Fail & \NA & \NA & \NA & \Fail & \NA & \NA & \NA\\
\path{lcm1__property.c.t2} & \Fail & \Fail & \NA & \NA & \Fail & \Fail & \NA & \NA\\
\path{lcm1.c.t2} & \NA & \NA & \Fail & \Fail & \NA & \NA & \Fail & \Fail\\
\path{lcm2.c.t2} & \Strong{4}{19.64} & \Strong{4}{21.34} & \Fail & \Fail & \Strong{4}{19.64} & \Strong{4}{21.34} & \Fail & \Fail\\
\path{LogMult.c.t2} & \Fail & \Strong{4}{7.07} & \Strong{2}{2.73} & \Fail & \Fail & \Strong{4}{7.07} & \Strong{2}{2.73} & \Fail\\
\path{mannadiv.c.t2} & \Strong{4}{7.13} & \Strong{3}{214.33} & \Fail & \Fail & \Strong{4}{7.13} & \Strong{3}{214.33} & \Fail & \Fail\\
\path{ps2-ll.c.t2} & \Strong{3}{2.69} & \Lazy{2}{326.90} & \Strong{1}{2.58} & \NA & \Strong{3}{2.69} & \Lazy{2}{326.85} & \Strong{1}{2.58} & \NA\\
\path{ps3-ll.c.t2} & \Strong{2}{24.60} & \Strong{3}{7.73} & \Strong{1}{2.43} & \NA & \Strong{2}{24.60} & \Strong{3}{7.73} & \Strong{1}{2.43} & \NA\\
\path{ps4-ll.c.t2} & \Strong{3}{2.40} & \Strong{3}{8.73} & \Strong{1}{2.75} & \NA & \Strong{3}{2.40} & \Strong{3}{8.73} & \Strong{1}{2.75} & \NA\\
\path{ps5-ll.c.t2} & \Fail & \Lazy{1}{484.58} & \Strong{1}{2.58} & \NA & \Fail & \Fail & \Strong{1}{2.58} & \NA\\
\path{ps6-ll.c.t2} & \Strong{3}{2.52} & \Strong{2}{291.05} & \Lazy{1}{472.77} & \NA & \Strong{3}{2.52} & \Strong{2}{291.05} & \Lazy{1}{472.67} & \NA\\
\path{sqrt1-ll.c.t2} & \Strong{3}{8.06} & \Lazy{3}{233.35} & \Strong{1}{3.40} & \NA & \Strong{3}{8.06} & \Lazy{3}{256.72} & \Strong{1}{3.40} & \NA\\
\path{svcomp_ex1.c.t2} & \Strong{7}{5.78} & \Lazy{1}{576.06} & \NA & \NA & \Strong{7}{5.78} & \Lazy{1}{576.11} & \NA & \NA\\
\path{svcomp_ex2__property.c.t2} & \Fail & \NA & \NA & \NA & \Fail & \NA & \NA & \NA\\
\path{svcomp_ex2.c.t2} & \NA & \Lazy{9}{921.70} & \Strong{1}{62.09} & \Strong{8}{45.95} & \NA & \Fail & \Strong{1}{62.09} & \Strong{8}{45.95}\\
\path{svcomp_ex3a.c.t2} & \Lazy{3}{3.57} & \Strong{4}{5.24} & \Strong{1}{1.10} & \NA & \Lazy{3}{3.59} & \Strong{4}{5.24} & \Strong{1}{1.10} & \NA\\
\path{svcomp_ex3b.c.t2} & \Fail & \Lazy{2}{380.20} & \Strong{1}{1.78} & \NA & \Lazy{3}{52.68} & \Fail & \Strong{1}{1.78} & \NA\\
\path{svcomp_fermat.c.t2} & \Lazy{2}{1363.68} & \Strong{1}{3.28} & \Strong{1}{3.89} & \Lazy{5}{577.83} & \Fail & \Strong{1}{3.28} & \Strong{1}{3.89} & \Lazy{5}{577.93}\\
\end{longtable}

\newpage
\begin{longtable}{>{\RaggedRight\arraybackslash}p{0.24\textwidth}*{8}{>{\centering\arraybackslash}p{0.085\textwidth}}}
\caption{Per-instance WORST and STAB results for the probabilistic-assignment variant.}\label{tab:detail-assignment-backends}\\
\toprule
Benchmark & \multicolumn{4}{c}{WORST} & \multicolumn{4}{c}{STAB}\\
\cmidrule(lr){2-5}\cmidrule(lr){6-9}
 & RA & OV & RC & PR & RA & OV & RC & PR\\
\midrule
\endfirsthead
\toprule
Benchmark & \multicolumn{4}{c}{WORST} & \multicolumn{4}{c}{STAB}\\
\cmidrule(lr){2-5}\cmidrule(lr){6-9}
 & RA & OV & RC & PR & RA & OV & RC & PR\\
\midrule
\endhead
\midrule
\multicolumn{9}{r}{\emph{Continued on next page}}\\
\endfoot
\bottomrule
\endlastfoot
\path{bresenham-ll.c.t2} & \Strong{4}{7.74} & \Strong{4}{37.31} & \NA & \NA & \Strong{4}{7.74} & \Strong{4}{37.31} & \NA & \NA\\
\path{ChawdharyCookGulwaniSagivYang-ESOP2008-aaron12_true-termination.c.t2} & \NA & \Strong{3}{148.89} & \NA & \NA & \NA & \Strong{3}{148.89} & \NA & \NA\\
\path{cohencu-ll.c.t2} & \Strong{3}{85.82} & \Strong{3}{7.77} & \NA & \NA & \Strong{3}{85.82} & \Strong{3}{7.77} & \NA & \NA\\
\path{cohendiv-ll.c.t2} & \Fail & \Lazy{1}{570.17} & \Strong{1}{436.06} & \NA & \Fail & \Lazy{1}{570.17} & \Strong{1}{436.06} & \NA\\
\path{ComplInterv.c.t2} & \NA & \Strong{3}{3.28} & \NA & \NA & \NA & \Strong{3}{3.28} & \NA & \NA\\
\path{dijkstra-u.c.t2} & \Lazy{2}{588.06} & \Lazy{6}{804.80} & \Lazy{4}{594.20} & \NA & \Fail & \Lazy{6}{804.92} & \Lazy{4}{594.27} & \NA\\
\path{divbin__property.c.t2} & \Fail & \NA & \NA & \NA & \Fail & \NA & \NA & \NA\\
\path{divbin.c.t2} & \NA & \Strong{8}{41.67} & \NA & \NA & \NA & \Strong{8}{41.67} & \NA & \NA\\
\path{DoubleNeg.c.t2} & \Strong{2}{1.04} & \Strong{2}{2.27} & \NA & \NA & \Strong{2}{1.04} & \Strong{2}{2.27} & \NA & \NA\\
\path{egcd-ll.c.t2} & \NA & \Fail & \NA & \NA & \NA & \Fail & \NA & \NA\\
\path{egcd2-ll.c.t2} & \Strong{4}{85.64} & \NA & \NA & \Fail & \Strong{4}{85.64} & \NA & \NA & \Fail\\
\path{egcd3-ll.c.t2} & \Lazy{7}{166.36} & \NA & \NA & \Fail & \Fail & \NA & \NA & \Fail\\
\path{Factorial__property.c.t2} & \NA & \Fail & \NA & \NA & \NA & \Fail & \NA & \NA\\
\path{Factorial.c.t2} & \NA & \NA & \NA & \Lazy{2}{6.93} & \NA & \NA & \NA & \Lazy{2}{6.94}\\
\path{geo1-ll.c.t2} & \Strong{3}{3.32} & \Lazy{1}{519.72} & \Strong{1}{3.56} & \Strong{2}{351.38} & \Strong{3}{3.32} & \Lazy{1}{519.68} & \Strong{1}{3.56} & \Strong{2}{351.38}\\
\path{geo2-ll.c.t2} & \Strong{3}{3.02} & \Lazy{1}{456.98} & \Strong{1}{2.78} & \Strong{2}{5.19} & \Strong{3}{3.02} & \Lazy{1}{456.97} & \Strong{1}{2.78} & \Strong{2}{5.19}\\
\path{geo3-ll.c.t2} & \Strong{3}{5.42} & \Lazy{1}{497.06} & \NA & \NA & \Strong{3}{5.42} & \Lazy{1}{497.01} & \NA & \NA\\
\path{hard-ll__property.c.t2} & \NA & \Lazy{6}{637.55} & \NA & \NA & \NA & \Fail & \NA & \NA\\
\path{hard-ll.c.t2} & \Fail & \NA & \NA & \NA & \Fail & \NA & \NA & \NA\\
\path{lcm1__property.c.t2} & \Fail & \Fail & \NA & \NA & \Fail & \Fail & \NA & \NA\\
\path{lcm1.c.t2} & \NA & \NA & \Fail & \Fail & \NA & \NA & \Fail & \Fail\\
\path{lcm2.c.t2} & \Strong{4}{19.60} & \Strong{4}{21.48} & \Fail & \Fail & \Strong{4}{19.60} & \Strong{4}{21.48} & \Fail & \Fail\\
\path{LogMult.c.t2} & \Strong{3}{1.65} & \Strong{4}{10.62} & \Strong{2}{2.61} & \Fail & \Strong{3}{1.65} & \Strong{4}{10.62} & \Strong{2}{2.61} & \Fail\\
\path{mannadiv.c.t2} & \Strong{4}{7.42} & \Strong{4}{176.93} & \Fail & \Fail & \Strong{4}{7.42} & \Strong{4}{176.93} & \Fail & \Fail\\
\path{ps2-ll.c.t2} & \Strong{3}{2.62} & \Strong{3}{10.20} & \Strong{1}{3.11} & \NA & \Strong{3}{2.62} & \Strong{3}{10.20} & \Strong{1}{3.11} & \NA\\
\path{ps3-ll.c.t2} & \Strong{3}{2.23} & \Lazy{1}{591.87} & \Strong{1}{2.36} & \NA & \Strong{3}{2.23} & \Lazy{1}{591.86} & \Strong{1}{2.36} & \NA\\
\path{ps4-ll.c.t2} & \Strong{3}{2.29} & \Lazy{1}{467.02} & \Strong{1}{2.05} & \NA & \Strong{3}{2.29} & \Lazy{1}{467.02} & \Strong{1}{2.05} & \NA\\
\path{ps5-ll.c.t2} & \Strong{3}{2.32} & \Strong{3}{8.42} & \Strong{1}{2.15} & \NA & \Strong{3}{2.32} & \Strong{3}{8.42} & \Strong{1}{2.15} & \NA\\
\path{ps6-ll.c.t2} & \Strong{3}{2.36} & \Lazy{1}{497.62} & \Strong{1}{2.15} & \NA & \Strong{3}{2.36} & \Lazy{1}{497.59} & \Strong{1}{2.15} & \NA\\
\path{sqrt1-ll.c.t2} & \Strong{3}{6.00} & \Strong{3}{18.96} & \Strong{1}{3.37} & \NA & \Strong{3}{6.00} & \Strong{3}{18.96} & \Strong{1}{3.37} & \NA\\
\path{svcomp_ex1.c.t2} & \Strong{7}{4.45} & \Lazy{1}{413.83} & \NA & \NA & \Strong{7}{4.45} & \Lazy{1}{413.82} & \NA & \NA\\
\path{svcomp_ex2__property.c.t2} & \Fail & \NA & \NA & \NA & \Fail & \NA & \NA & \NA\\
\path{svcomp_ex2.c.t2} & \NA & \Lazy{9}{1124.15} & \Strong{1}{7.24} & \Strong{8}{28.88} & \NA & \Lazy{9}{1124.15} & \Strong{1}{7.24} & \Strong{8}{28.88}\\
\path{svcomp_ex3a.c.t2} & \Lazy{3}{1.51} & \Strong{4}{2.67} & \Strong{1}{1.07} & \NA & \Lazy{3}{1.51} & \Strong{4}{2.67} & \Strong{1}{1.07} & \NA\\
\path{svcomp_ex3b.c.t2} & \Lazy{3}{2.71} & \Strong{3}{26.20} & \Strong{1}{1.60} & \NA & \Lazy{3}{2.71} & \Strong{3}{26.20} & \Strong{1}{1.60} & \NA\\
\path{svcomp_fermat.c.t2} & \Strong{13}{657.03} & \Strong{1}{3.24} & \Strong{1}{3.85} & \Lazy{5}{566.50} & \Strong{13}{657.03} & \Strong{1}{3.24} & \Strong{1}{3.85} & \Lazy{5}{566.42}\\
\end{longtable}

\endgroup

\end{document}